\documentclass[letterpaper,journal]{IEEEtran}
\usepackage{amsmath,amsfonts}
\usepackage{algorithm}
\usepackage{algorithmic}

\usepackage{color}
\usepackage{xcolor}
\usepackage{array}
\usepackage[caption=false,font=normalsize,labelfont=sf,textfont=sf]{subfig}
\usepackage{textcomp}
\usepackage{stfloats}
\usepackage{url}
\usepackage{verbatim}
\usepackage{graphicx}
\usepackage{cite}
\usepackage{amsmath}
\usepackage{amsthm}
\usepackage{amssymb}
\usepackage{enumitem}
\usepackage{subfig}
\usepackage{booktabs} 
\usepackage{multirow} 
\newcommand{\q}[1]{\textcolor{black}{#1}}
\newcommand{\revi}[1]{\textcolor{black}{#1}} 
\newcommand{\revii}[1]{\textcolor{black}{#1}} 

\usepackage[utf8]{inputenc}   
\usepackage[T1]{fontenc}      
\usepackage{CJKutf8}

\newtheorem{theorem}{Theorem}

\newtheorem{lemma}{Lemma}
\newtheorem{definition}{Definition}

\usepackage{mathptmx}
\DeclareMathAlphabet{\mathcal}{OMS}{cmsy}{m}{n}
\DeclareMathAlphabet{\mathrm}{OT1}{bch}{m}{n}
\DeclareMathAlphabet{\mathit}{OT1}{bch}{m}{it}

\begin{document}
\begin{CJK}{UTF8}{gbsn}

\title{Joint Optimization of DNN Model Caching and Request Routing in Mobile Edge Computing}
\author{
Shuting Qiu,
Fang Dong,~\IEEEmembership{Member,~IEEE,}
Siyu Tan,
Ruiting Zhou,~\IEEEmembership{Member,~IEEE,}
Dian Shen,~\IEEEmembership{Member,~IEEE,}
Patrick P. C. Lee,~\IEEEmembership{Senior~Member,~IEEE}
and 
Qilin Fan,~\IEEEmembership{Member,~IEEE,}

\thanks{
This work is supported by National Natural Science Foundation of China under Grants, No. 62232004, Jiangsu Provincial Frontier Technology Research and Development Program under Grant BF2024070, Shenzhen Science and Technology Program under Grant KJZD20240903100814018, Jiangsu Provincial Key Laboratory of Network and Information Security under Grants No.BM2003201, Key Laboratory of Computer Network and Information Integration of Ministry of Education of China under Grants No.93K-9, Research Grants Council of Hong Kong (GRF 14201523), and partially supported by Collaborative Innovation Center of Novel Software Technology and Industrialization. (Corresponding author: Fang Dong.)

An earlier version of this paper appeared at IEEE INFOCOM 2025~\cite{cocar} [DOI: 10.1109/INFOCOM55648.2025.11044457].
\revi{In this extended version, we address online scenarios with unpredictable future requests. We propose CoCaR-OL, a framework that adaptively adjusts dynamic DNN caching strategies. We further conduct new experiments to validate its effectiveness.}
}

\thanks{
Shuting Qiu, Fang Dong, Siyu Tan, Ruiting Zhou and Dian Shen
are with the School of Computer Science and Engineering, Southeast University, Nanjing 211189, China (email: qiushuting@seu.edu.cn; fdong@seu.edu.cn; sytan@seu.edu.cn; ruitingzhou@seu.edu.cn; dshen@seu.edu.cn).
}
\thanks{
Patrick P. C. Lee is with the Department
of Computer Science and Engineering, The Chinese University of Hong Kong, Hong Kong (email: pclee@cse.cuhk.edu.hk).
}
\thanks{
Qilin Fan is with the School of Big Data and Software Engineering, Chongqing University, Chongqing 400044, China (email: fanqilin@cqu.edu.cn).
}
}

\maketitle

\begin{abstract}
\label{abstract}
Mobile edge computing (MEC) can pre-cache deep neural networks (DNNs) near end-users, providing low-latency services and improving users' quality of experience (QoE). However, caching all DNN models at capacity-limited edge servers is difficult, and the impact of model loading time on QoE remains underexplored.  We explore dynamic DNNs by disassembling a complete DNN model into interrelated submodels to enable fine-grained joint optimization of submodel caching and request routing to balance inference precision and loading latency. 
In this paper, we study the joint dynamic model caching and request routing problem in MEC networks, aiming to maximize user request inference precision under constraints of server resources, latency, and model loading time. 
\revii{We propose CoCaR, an offline algorithm based on linear programming and random rounding that optimizes joint decisions with a provable performance bound. Furthermore, we develop an online extension, CoCaR-OL, to adapt to dynamic and unpredictable request patterns.}
The simulation results demonstrate that CoCaR improves the average inference precision for user requests by 40.1\% over state-of-the-art baselines. 
In addition, CoCaR-OL achieves an improvement of at least 32.3\% in users' QoE over competitive baselines.

\end{abstract}

\begin{IEEEkeywords}

Mobile edge computing, dynamic DNN, model caching, request routing, joint optimization

\end{IEEEkeywords}
\section{Introduction}
\label{introduction}

Recently, machine learning methods, especially deep neural networks (DNNs), have been transforming various application domains \cite{hu,wtf1,ton-dnn2}, such as computer vision \cite{cv}, speech recognition \cite{sr}, and autonomous driving \cite{sdc}. To promote efficient task processing, mobile edge computing (MEC) \cite{mec,drl2,ton-mec1,ton-mec2} has emerged as a promising paradigm to provide end-users with low-latency computation, caching, and transmission capabilities \cite{mec1,math1,ton-tan}. In MEC, DNN models can be pre-cached at base stations (BSs) with edge servers, enabling rapid responses to user requests \cite{mec2}, \cite{mec3}. However, unlike cloud servers, due to the limited resources of a BS, only a few popular DNN models can be cached simultaneously to serve users \cite{zhao}, \cite{mec4}.
Additionally, routing requests to BSs with relevant cached models is essential for user coverage and system efficiency, and is hence critical to users' quality of experience (QoE) in multi-edge scenarios \cite{chu,ton-route,wtf2}. Since model caching and request routing are tightly coupled, jointly optimizing them is essential for users' QoE.

Existing joint optimization schemes for model caching and request routing mainly focus on caching complete original models at resource-limited BSs \cite{poularakis,math2,drl1,yao,tsc1}, leading to inefficient resource \q{utilization} and poor user QoE \cite{smart-momeps}.
Fortunately, dynamic DNNs \cite{ddnn,ton-dnn} offer a promising solution by enabling flexible adjustments to the model structure. In this paper, we define dynamic DNNs as a model versioning framework in which a base model is partitioned into multiple submodels of varying depths and precisions for fine-grained caching. Specifically, we disassemble a complete original DNN model into a set of interrelated submodels by dividing the feature extraction layer, allowing flexible switching among them by adjusting the number of layers. This approach enables fine-grained resource allocation in a multi-server caching environment, thereby distinguishing our method from device-edge split computing \cite{split1,split2}, which partitions a single inference task across two locations for collaborative execution. Consequently, our use of dynamic DNNs enables more flexible model caching and request routing, allowing BS resources to be used in a more fine-grained manner.


However, incorporating dynamic DNNs complicates the joint caching and routing decisions, leading to two primary challenges.
First, caching schemes for submodels of the same dynamic DNN at a BS are mutually exclusive, permitting only one submodel to be cached at a time. This necessitates determining not only which DNN models to cache but also which submodel of these dynamic DNNs to cache, thereby increasing the complexity of the variable-coupled joint optimization problem. Second, inference requests for different DNN models arrive in real-time, requiring timely adjustments to caching schemes \cite{qiu}. Moreover, caching models at BSs incurs loading latency \cite{load}, and different caching schemes for dynamic DNNs yield different model-valid service times and inference precision, thereby affecting users' QoE. Thus, the adverse impact of model loading time on QoE must be considered and minimized. Although some studies consider model loading time \cite{load,load1}, they use it only to refine solutions rather than to inform caching decisions, leaving existing approaches insufficiently adaptive to dynamic, resource-constrained MEC environments.


In this paper, we study a \textbf{J}oint \textbf{D}ynamic model \textbf{C}aching and request \textbf{R}outing problem (JDCR) with the goal of maximizing user request inference precision.
Based on linear programming (LP) and randomized rounding, we first propose a novel offline multi-edge \textbf{Co}llaborative dynamic model \textbf{Ca}ching and request \textbf{R}outing optimization algorithm (CoCaR), which provides an approximate optimal solution to the JDCR problem with theoretical guarantees.
Subsequently, we extend CoCaR to CoCaR-OL to adapt to online scenarios where user requests are hard to predict.
Specifically, we divide time into consecutive observation windows. When making decisions in each window, we first consider the BS resource and latency constraints, as well as the impact of model loading time on QoE due to the switching of caching schemes.
Then, based on BS caching results from the previous window, we utilize the capability of rapid switching between submodels of the same dynamic DNN to flexibly adjust the caching scheme, thereby increasing the valid service time of the model.
Additionally, by leveraging the fine-grained resource utilization feature of dynamic DNNs, BS resources can be used more efficiently, enabling the caching of more DNN models to satisfy a greater number of user demands.
The contributions of this paper are as follows:

\begin{itemize}[leftmargin=*]
    
    \item We formulate the JDCR problem in MEC networks. The goal is to maximize the total inference precision for user requests while adhering to constraints on BS resources, latency, and model loading time. As this JDCR problem is a non-convex, nonlinear integer programming problem that poses significant challenges for resolution, we perform equivalent transformations and relaxations.
    
    \item We propose a novel \q{offline} algorithm, CoCaR, based on LP and random rounding to solve the JDCR problem. Specifically, CoCaR can leverage dynamic DNNs to provide a more flexible and fine-grained caching scheme.
    Theoretical analyses show that CoCaR has strict performance guarantees and achieves an approximation ratio of $(1-\sqrt{\frac{4\ln|\mathcal{H}|}{\mathbb{P}^{\dagger}}})^2$ to the optimal integer solution, where $|\mathcal{H}|$ represents the total number of submodels, and $\mathbb{P}^{\dagger}$ is the objective value of the optimal fractional solution.
    \item We further extend CoCaR to handle online scenarios where future user requests are difficult to predict in advance and both caching and routing decisions must be made in real-time. Thus, we propose CoCaR-OL, an online variant that determines dynamic DNN model caching strategies based on expected future gains estimated from historical user request patterns. It further adopts a greedy algorithm for routing decisions, enabling effective adaptation to dynamic and unpredictable demands.

    \item Extensive simulations \q{validate} the effectiveness of the proposed CoCaR \q{and CoCaR-OL}. The results show that CoCaR achieves at least a 40.1\% improvement in average inference precision and a 42.1\% increase in cache hit rate over four state-of-the-art algorithms, utilizing at least 86\% of BS resources.
    In online scenarios, CoCaR-OL attains at least a 1.71$\times$ improvement in users' QoE and a 1.73$\times$ increase in cache hit rate over the baseline.       
\end{itemize}

\revi{The remainder of this paper proceeds as follows. Sec.~II reviews related work. Sec.~III describes the background and motivation. Sec.~IV introduces the system model and formulates the joint caching and routing problem. Sec.~V focuses on the offline scenario, proposing the CoCaR algorithm and analyzing its approximation ratio. Sec.~VI extends the study to online scenarios and proposes the CoCaR-OL algorithm based on expected future gains. Sec.~VII presents the simulation results, and Sec.~VIII concludes the paper.}

\section{\revi{Related work}}
\label{related work}

\textbf{Caching and Routing.}
Since service caching and request routing are critical in MEC, jointly optimizing them to guarantee users’ QoE remains a major challenge.
Zhang \textit{et al.} \cite{{math2}} jointly consider service caching, computation offloading, transmission, and computation resource allocation to minimize overall computation and latency costs.
Chu \textit{et al.} \cite{chu} study how to maximize users' QoE by jointly optimizing service caching, resource allocation, and task offloading decisions.
Zhao \textit{et al.} \cite{zhao} study how to efficiently offload dependent tasks to edge nodes with limited service caches and design an efficient algorithm based on convex programming to solve this problem.
Bi \textit{et al.} \cite{bi} consider the joint optimization of service cache locations, computation offloading decisions, and system resource allocation on a single edge server.
\revii{In video caching, Li \textit{et al.} \cite{video} optimize edge server storage allocation for various bitrate versions by leveraging the rate-distortion characteristics of video content. 
Additionally, several works in serverless computing have addressed the latency issues in containerized deployments. For example, Pagurus \cite{serverless1} mitigates cold starts by enabling cross-function container transformation and sharing idle warm containers, while RainbowCake \cite{serverless2} proposes a hierarchical pre-warming and residency mechanism to facilitate fine-grained, layered caching and cross-function sharing and reduce cold-start latency.}

\revii{However, these existing studies predominantly focus on caching complete, original services or function instances at BSs. In contrast, we consider the structural dependencies among DNN submodels, treating them as correlated rather than independent cache objects, thereby enabling more efficient utilization of BS resources. Moreover, most prior works neglect the impact of service loading time on joint caching and routing decisions. In practical MEC environments where user request patterns are highly dynamic, overlooking loading latency can severely degrade users’ QoE.}

\revii{Some studies  (e.g., \cite{ton-mec2,mec4}) address online service placement via the Lyapunov framework, which takes reactive, drift-plus-penalty optimization based on current system states to ensure long-term stability. In contrast, our online approach (CoCaR-OL) takes a predictive, expected-future-gain approach to proactively manage the multi-time-slot download latency associated with dynamic submodel switching. We also address the switching cost that couples cache and routing decisions, which cannot be directly addressed by the standard Lyapunov drift bounds.} 


\textbf{Dynamic DNNs.}
Dynamic DNNs can meet various requirements under different resource constraints by adjusting the number of channels and layers.
MutualNet \cite{mutualnet} and MSDNet \cite{msdnet} achieve adaptive precision-latency tradeoffs by varying DNN network width and depth, respectively.
Dynamic-OFA \cite{ddnn} provides an efficient architecture for GPUs and CPUs by using a shared backbone model that adjusts DNN width, depth, filter sizes, and input resolution.
SubFlow \cite{subflow} dynamically builds and executes DNN sub-networks for flexible time-bound inference and training.
Smart-MOMEPS \cite{smart-momeps} leverages a multi-exit mechanism and dynamic service migration to alleviate DNN performance degradation.
SN-Net \cite{sn-net} generates numerous new stitched networks by combining pre-trained models from a model family, enabling adaptation to different resource constraints. 
ESTA \cite{esta}, on the other hand, uses efficient parameter fine-tuning and training-time gradient statistics to assemble stitched models of various sizes from existing pre-trained DNNs, achieving a stable balance between precision and efficiency.

Unlike existing studies, we introduce dynamic DNNs into edge caching by dividing a complete DNN into multiple interrelated, switchable submodels. 
This approach provides a novel solution for fine-grained, resource-efficient, and flexible edge-based model caching.

\section{Background and Motivation}
\label{motivation}

\begin{figure}[!t]
\centering
\includegraphics[width=.45\textwidth]{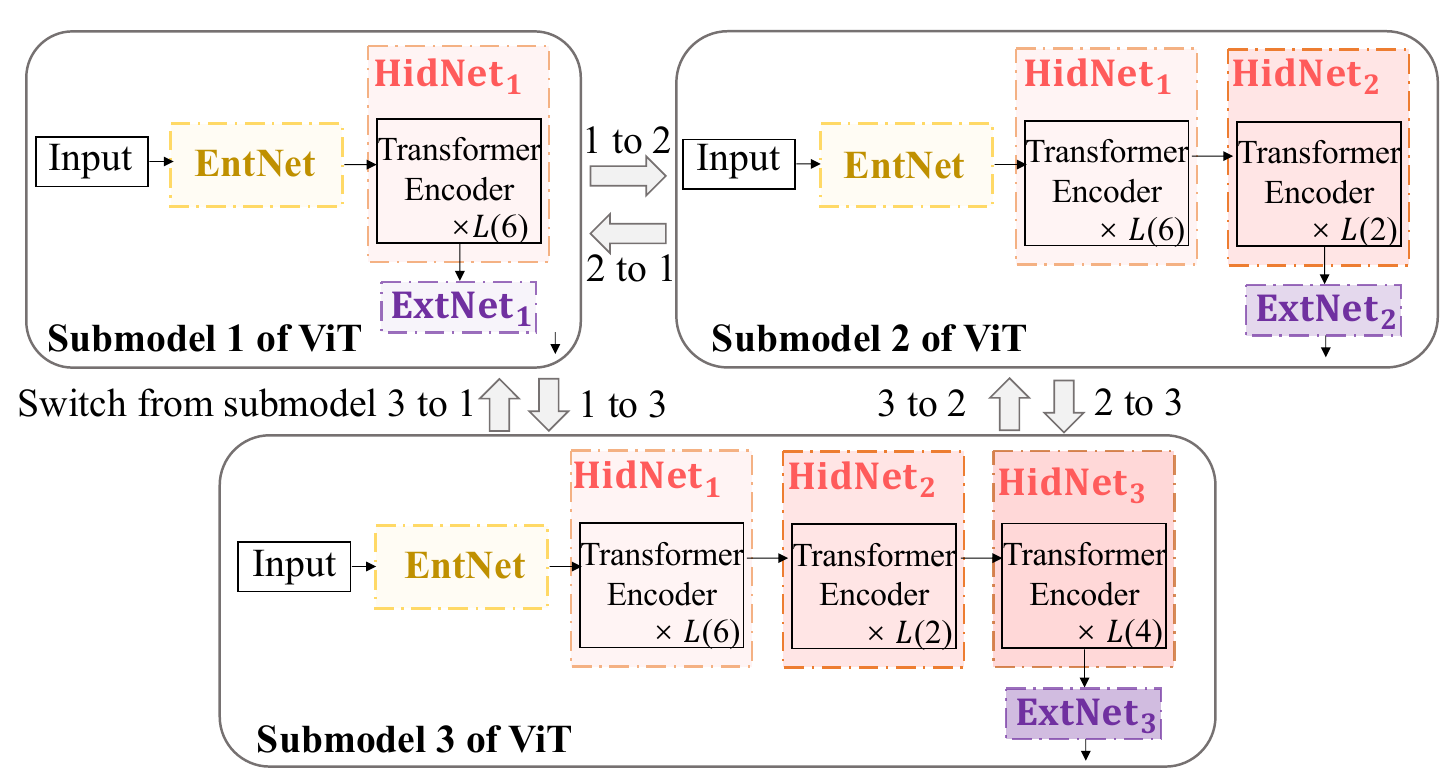}
\caption{Submodels division and switching of ViT. When switching from submodel 1 to submodel 2 of ViT,  we only need to remove $\text{ExtNet}_1$ of submodel 1 and then connect $\text{HidNet}_2$ and $\text{ExtNet}_2$ to form submodel 2.} 
\label{fig:dynamic_dnn}
\end{figure}

\revii{
\textbf{Dynamic DNN Architecture}. As shown in Fig.~\ref{fig:dynamic_dnn}, using the Vision Transformer (ViT) \cite{vit} as a representative example, we divide a standard DNN into three primary components: the entry network (EntNet), the hidden network (HidNet), and the exit network (ExtNet). To support varying levels of inference precision, the original DNN is partitioned into multiple submodels of different sizes. This is achieved by dividing HidNet into segments of varying depths (e.g., $L(k)$ denotes a segment with $k$ Transformer encoder layers), where each submodel comprises multiple HidNet segments and is equipped with its own uniquely trained ExtNet. The arrows in Fig.~\ref{fig:dynamic_dnn} indicate the valid transitions for submodel switching, which can be executed rapidly by appending or pruning specific HidNet segments and their associated ExtNets.
}

{\color{black}
\textbf{Motivating Example.} To demonstrate the advantages of the dynamic DNN architecture, we compare a traditional caching scheme that strictly utilizes complete models (i.e., a static DNN scheme) against our proposed dynamic DNN-based scheme at a single BS. Aside from the specific parameters outlined below, the example utilizes the experimental configurations detailed in Sec.~\ref{para}.

As shown in Fig.~\ref{fig:motivation_example}, there are two DNN model types, A and B, each divided into three submodels of different sizes and precisions.  In each observation window, 100 users uniformly initiate model inference requests. The duration of each observation window is set to 5\,s. The BS has a strict cache capacity of 2\,GB.  Cached models are loaded into the BS memory at the start of each window and can only begin providing valid inference services once all required layers of the submodel are fully resident in memory. 

\begin{figure}[t]
\centering
\includegraphics[width=.48\textwidth]{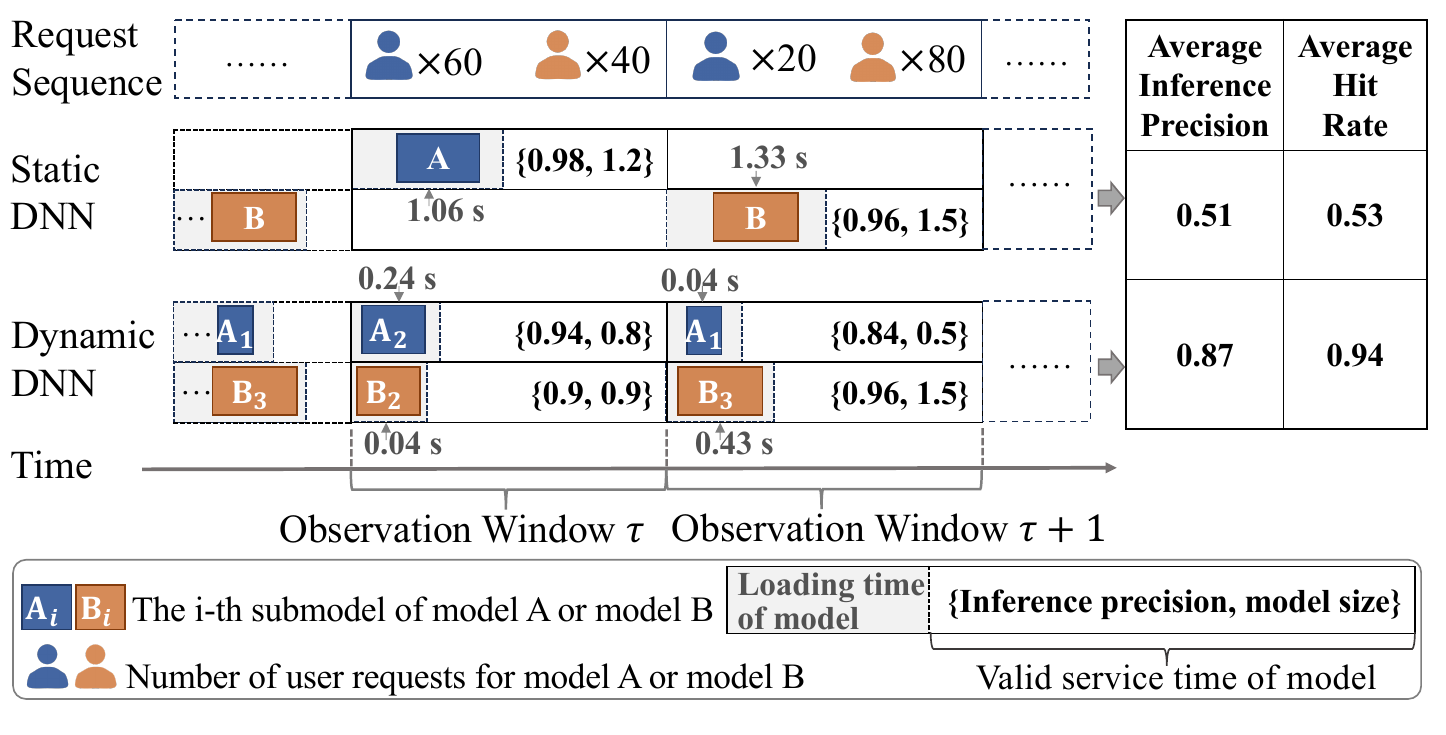}
\vspace{-6pt}
\caption{Examples of static DNN and dynamic DNN schemes.} 
\label{fig:motivation_example}
\end{figure}

Suppose that in the observation window $\tau$, 60 users request Model~A and 40 users request Model~B. In the subsequent window $\tau+1$, the request distribution shifts to 20 users for Model~A and 80 users for Model~B. Next, we evaluate how the two schemes perform across both windows.
\begin{itemize}[leftmargin=*]
\item 
{\em Static DNN scheme (complete models only):} In window $\tau$, the BS caches the complete original Model~A (precision 0.98, size 1.2\,GB), which takes 1.06\,s to load. It must drop all requests for Model~B as caching both complete models would require 2.7\,GB, exceeding the 2\,GB limit. In window~$\tau+1$, to serve the new surge in Model~B requests, the static scheme drops Model~A and attempts to load the complete Model~B (precision 0.96, size 1.5\,GB) from scratch, incurring a high loading latency of 1.33\,s.
\item
{\em Dynamic DNN scheme (submodel switching):} In window $\tau$, the dynamic scheme successfully fits both models into the cache by utilizing smaller submodels (e.g., caching submodel~2 for Model~B). In window $\tau+1$, it can leverage the previously cached state: instead of loading a complete model from scratch, it upgrades Model~B from submodel~2 to a higher precision state (submodel~3: precision 0.96, size 1.5\,GB). By loading only a submodel, this switch takes only 0.43\,s. This leaves 0.5\,GB of available capacity, which perfectly accommodates submodel~1 of Model~A (precision 0.84, size 0.5\,GB), taking only 0.04\,s to load.
\end{itemize}

With dynamic submodels, the BS successfully accommodates both models within the strict 2\,GB capacity constraint and significantly reduces loading latency. To quantify the benefit, we consider two metrics: the average inference precision ($P_{avg}$) and the average cache hit rate ($H_{avg}$). In each observation window $\tau$, $P_{avg}$ and $H_{avg}$ are calculated as
$P_{avg} = \frac{\sum_{h \in \mathcal{H}_c} \left( \lfloor \frac{u_h}{|\tau|} \times (|\tau| - l_h) \rfloor \times p_h \right)}{U_{total}}$, and 
$H_{avg} = \frac{\sum_{h \in \mathcal{H}_c} \left( \lfloor \frac{u_h}{|\tau|} \times (|\tau| - l_h) \rfloor \right)}{U_{total}}$, where $\mathcal{H}_c$ denotes the set of cached submodels, $u_h$ is the number of user requests for submodel $h$ in window $\tau$ with duration $|\tau|$, $l_h$ and $p_h$ denote the loading time and inference precision of submodel $h$, respectively, and $U_{total}$ is the total number of user requests in window $\tau$. 

For example, calculating the average inference precision for the static DNN across both windows incorporates its performance serving Model~A in window~$\tau$ and Model~B in window $\tau+1$: $(\lfloor(60/5) \times (5-1.06)\rfloor \times 0.98 + \lfloor(80/5) \times (5-1.33)\rfloor \times 0.96) / 200 = 0.51$. Ultimately, the average inference precision and average hit rate of the dynamic DNN scheme are 0.87 and 0.94, respectively (absolute increases of 0.36 and 0.41 over the static DNN baseline). This underscores the need for novel caching and routing optimization designs for the dynamic DNN scheme in MEC environments.
}

\section{System model and problem formulation} 
\label{system model}

\begin{figure}[t]
\centering
\includegraphics[width=.45\textwidth]{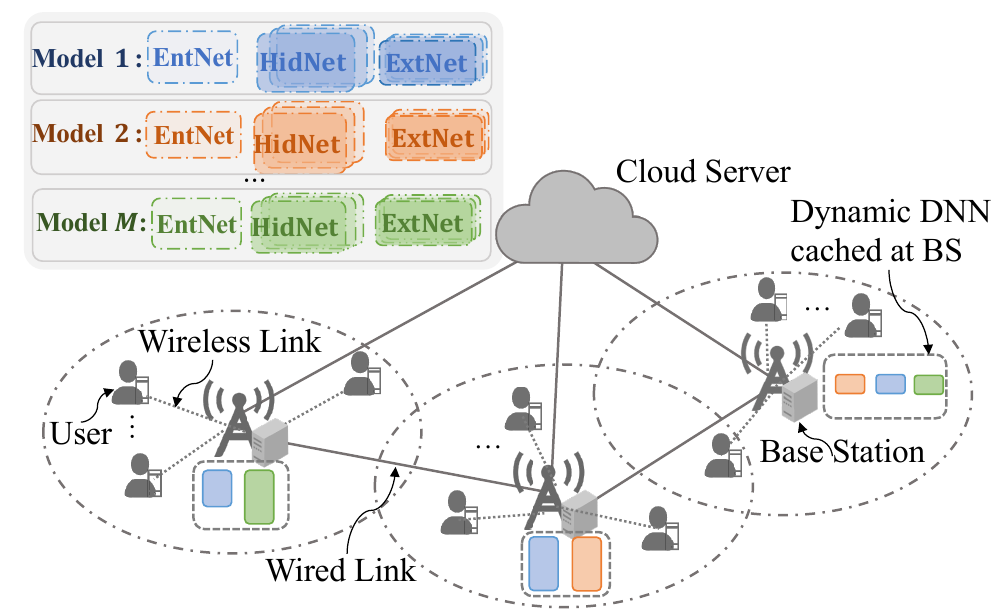}
\caption{System model. The same color indicates the dynamic DNN associated with a model type, while different lengths of that color represent different submodels of the dynamic DNN.} 
\label{fig:system_model}
\end{figure}

\subsection{System Overview} 
As shown in Fig. \ref{fig:system_model}, 
we consider an MEC system consisting of $\mathcal{N} = \{1, 2, \cdots, N\}$ BSs equipped with edge servers.
These BSs have caching and computing capabilities and can communicate with each other through high-speed wired networks. They can provide some of the $M$ services for end users \q{$\mathcal{U} = \{1, 2, \cdots, U\}$}. 
We assume that users are randomly and uniformly distributed within the BS coverage area, and denote the closest BS to each user $u$ (the home BS) by $\hat{n}_u$.
We focus on joint model caching and request routing in MEC networks.
Let \q{$\mathcal{M} = \{1, 2, \cdots, M\}$} represent the set of different DNN model types, and $m_u \in \mathcal{M}$ the DNN model inference request generated by user $u$.
Let \q{$\mathcal{H}(m) = \{h^{m}_0, h^{m}_1, \cdots, h^{m}_{H(m)}\}$} denote the set of submodels \q{$\{h^{m}_j\}_{j=1}^{H(m)}$} of the dynamic DNN associated with \q{$m$}, along with an empty submodel \q{$h^{m}_0$} 
\footnote{Here, $h^{m}_0$ is introduced for modeling convenience, with no impact on BS resource usage or model inference precision.}.
Then, denote by \q{$\mathcal{H} = \bigcup_{m \in \mathcal{M}} \mathcal{H}(m)$} the set of all submodels, and $|\mathcal{H}|$ the total number of submodels.
For each model type \q{$m$}, a BS can cache at most one submodel concurrently to provide service.
\q{Furthermore, we define a partial order $(\mathcal{H}, \preceq)$ over the set of submodels $\mathcal{H}$, where each $\mathcal{H}(m)$ is linearly ordered, and elements from different $\mathcal{H}(m)$ are incomparable.
That is, for any $h_i, h_j \in \mathcal{H}$, we have $h_i \preceq h_j$ if and only if $h_i, h_j \in \mathcal{H}(m)$, and $h_i$ is a submodel no larger than $h_j$.
Similarly, $h_i \succ h_j$ holds if and only if $h_i$ and $h_j$ belong to the same $\mathcal{H}(m)$, and $h_i$ is a larger submodel than $h_j$.}

Considering the dynamic nature of the MEC network, time is divided into discrete time slots, denoted as $\mathcal{T} = \{1, 2, \ldots, t, \ldots, T\}$. 
Multiple consecutive time slots form an observation window, denoted as $\Gamma = \{1, 2, \ldots, \tau, \ldots, |\Gamma|\}$. 
Each observation window updates the caching and routing decisions based on user requests. If user $u$'s request is routed to the target BS that caches the submodel of the dynamic DNN of $m_u$, it is called a hit, and model inference is performed to obtain a result with corresponding precision; otherwise, it is a miss, and the user $u$'s request will be sent to the cloud with an inference precision of 0 at the edge.

Next, we focus on the joint offline caching and routing decisions within an observation window $\tau$. For brevity, we omit $\tau$ when there is no ambiguity.

\subsection{Model Caching}
It is crucial to decide which DNN models to cache on resource-limited BSs.
Let \revii{the} binary variable \q{$ x_{n,h} \in \{0, 1\}$} be the model caching decision for submodel $h$ in BS $n$ at window $\tau$,
where $x_{n,h}=1$ if submodel $h$ has been cached at BS $n$, and $ x_{n,h}=0$ otherwise.
Each BS is required to cache a submodel of the dynamic DNN associated with each model type:
\begin{equation}
    \label{con.x}
    \sum_{h \in \mathcal{H}(m)}  x_{n,h} = 1, \forall n \in \mathcal{N}, m \in \mathcal{M},
\end{equation}
where if \q{$x_{n,h^{m}_0}=1$}, it indicates that the empty submodel of model type \q{$m$} is cached (which is equivalent to not caching the model type \q{$m$}).

In addition, the total cached models cannot
exceed the memory capacity of BS $n$:
\begin{equation}
    \label{con.R}
    \sum_{h\in \mathcal{H}}  x_{n,h}\cdot r_h \le R_n, \forall n \in \mathcal{N},
\end{equation}
where $r_h$ is the required memory space to cache submodel $h$, and $R_n$ is the total memory capacity of BS $n$. 

\begin{figure}[t]
\centering
\includegraphics[width=.36\textwidth]{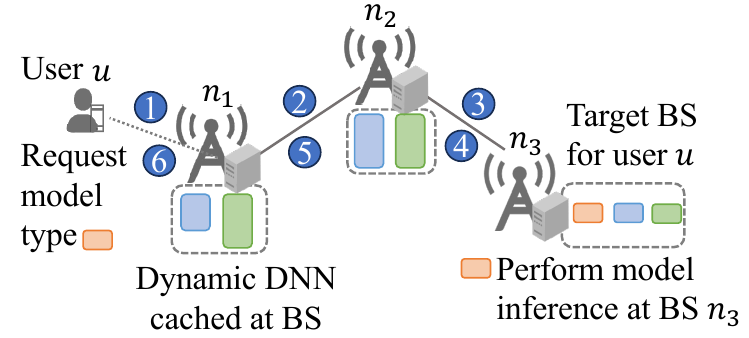}
\caption{Communication latency: Routing user $u$'s request involves wireless transmission from $u$ to home BS $n_1$, wired transmission from $n_1$ to target BS $n_3$, and a total of 6 hops from initiating the request to receive the inference result.} 
\label{fig:routing}
\end{figure}

\subsection{Request Routing}
\label{request routing model}

In general, there are many methods to pre-fetch user requests \cite{math2, ton-user}, which can then be routed to the target BS or the cloud for inference. Let $ y_{n,u} \in \{0, 1\}$ be the user request routing decision variable at window $\tau$, where  $ y_{n,u} =1$ if the request from user $u$ is routed to BS $n$; otherwise, $ y_{n,u} =0$. Let $d _u$, $ddl_u$, and $s_u$ denote the request data size of user $u$, the maximum perceived latency user $u$ can tolerate, and the initiation time of user $u$'s request in the observation window $\tau$, respectively. 

First, each user $u$'s request is routed to at most one BS, and requests that cannot be routed are sent to the cloud:
\begin{equation}
    \label{con.y}
    \sum_{n \in \mathcal{N}}  y_{n,u} \le 1, \forall u \in \mathcal{U}.
\end{equation}

\begin{table}[t]
\caption{List of Notations.}
\label{table:notation}
\vspace{-6pt}
\centering
\renewcommand{\arraystretch}{1.2} 
\begin{tabular}{l|l}
\toprule 
\textbf{Notation} & \textbf{Description} \\
\midrule
$\mathcal{N}, \mathcal{U}$                           & Set of base stations and users\\
$\mathcal{M}, \mathcal{H}$                           & Set of DNN model types and submodels \\
$\mathcal{H}(m)$                                     & Set of submodels associated with model type $m$\\
$\hat{n}_u$                                          & Closest BS to user $u$ (home BS)\\
$m_u$                                                & User $u$’s requested DNN model type $m$\\
$x_{n,h}$                                            & Caching variable for submodel $h$ at BS $n$ in window $\tau$ \\
$y_{n,u}$                                            & Routing variable for user $u$ to BS $n$ in window $\tau$\\
$d_u, ddl_u$                                         & User $u$'s data size and maximum tolerable latency\\
$s_u$                                                & User $u$'s request initiation time in window $\tau$ \\   
$R_n$                                                &Memory capacity of BS $n$  \\
$C_n$                                                &Computation capacity of BS $n$ \\
$r_h$                                                & Size of submodel $h$\\
$c_h$                                                & Flops of submodel $h$ per data unit\\
$\varphi_{n}$                                        & Wireless transmission rate of BS $n$\\ 
$r_{n', n}$                                          & Wired transmission rate between BS $n'$ and BS $n$\\ 
$\lambda_{u,n}$      
& Propagation latency for user $u$ to BS $n$ in window $\tau$\\
$T^{\text{off}}_u$                                          &User $u$'s communication latency in window $\tau$\\
$T^{\text{infer}}_u$                                        &User $u$'s inference latency in window $\tau$\\
$\mathbb{T}_{u}$                                     &User $u$'s total end-to-end inference latency in window $\tau$\\
$T^{\text{load}}_{n,m}$                                      &Load latency of model type $m$ at BS $n$ in window $\tau$\\
$p_h$                                                &Expected inference precision of submodel $h$\\
\bottomrule 
\end{tabular}
\end{table}

Second, users experience end-to-end inference latency from request initiation to receiving the result, which includes:
\begin{itemize}[leftmargin=*]
    \item \textit{Communication latency}. As shown in Fig. \ref{fig:routing}, the communication latency for routing user requests can be calculated as 
    $T^{\text{off}}_{u} = \sum_{n\in \mathcal{N}}  y_{n,u} \cdot (\frac{d_u}{\varphi_{\hat{n}_u}} + \frac{d_u}{r_{\hat{n}_u,n}} + \lambda_{u,n}), \forall  u \in \mathcal{U}$, 
    where $\varphi_{\hat{n}_u} = W_c log_2(1 + \frac{g_u E_{n,u}}{\mathbb{N}_0})$ is the wireless transmission rate, $g_u$ is user $u$'s transmission power, $E_{n,u}$ is the channel gain from user $u$ to BS $n$, $W_c$ is the channel bandwidth, and $\mathbb{N}_0$ is the noise power 
    \footnote{In this paper, we simplify the transmission model by assuming that BSs adopt a fixed transmission rate $\varphi_{\hat{n}_u}$ for all users $u$ \cite{chen2019}.}; 
    $r_{\hat{n}_u,n}$ is the average wired transmission rate between home BS $\hat{n}_u$ and target BS $n$; and $\lambda_{u,n}$ denotes the propagation latency of user $u$’s request routed to BS $n$, measured from its initiation to the reception of the inference result, which depends on the number of hops.
    
    \item \textit{Inference latency}. The inference latency of user $u$'s request at the target BS $n$ can be calculated as $T^{\text{infer}}_{u} = \sum_{n\in \mathcal{N}}\sum_{h\in \mathcal{H}(m_u)}  y_{n,u} \cdot  x_{n,h} \cdot \frac{c_h\cdot d_u}{C_n}, \forall u \in \mathcal{U}$, where $c_h$ denotes the computational flops required for model $h$ to process one unit of data request, and $C_n$ denotes the maximum computational capacity that BS $n$ can provide.
\end{itemize}

Thus, user $u$'s end-to-end inference latency $\mathbb{T}_{u}$ can be calculated as
$\mathbb{T}_{u} = T^{\text{off}}_{u} + T^{\text{infer}}_{u}, \forall u \in \mathcal{U}$. To ensure the QoE of end-users, $\mathbb{T}_{u}$ must not exceed the maximum perceived latency $ddl_u$ that user $u$ can tolerate:
\begin{equation}
    \label{con.T}
    \mathbb{T}_{u} \le ddl_u, \forall u \in \mathcal{U}.
\end{equation}

Additionally, the time required to load the DNN model into the BS's memory at each observation window cannot be ignored.
\q
{In offline scenarios, we can proactively pre-download the required DNN models from the cloud to the secondary storage of the BS based on the expected user requests in the upcoming window, and then load the models into the memory at the beginning of the next observation window.
}
Let $T^{\text{load}}_{n,m}$ denote the loading latency for caching the submodel of the dynamic DNN associated with model type $m$ at BS $n$.
When caching submodel $h$ of a model type that is already cached in the previous window, let $D_m^{swit}(h', h)$ denote the switching latency from submodel $h'$ to submodel $h$; 
Conversely, when caching a model that has not been cached before, let $D^{new}_{m}(h^m_0, h)$ represent the latency to add submodel $h$ of model $m$.
For the window $\tau$, the model loading latency of the model type $m$ at BS $n$ can be calculated as: 
\begin{equation}
    \begin{aligned}
    \label{T^dep}
    T^{\text{load}}_{n,m}  &= \sum_{h, h' \in {\{\mathcal{H}(m)\backslash {h^m_0}\}}} x_{n,h'}(\tau-1)\cdot D^{swit}_{m}(h', h) \cdot  x_{n,h} \\
    & + \sum_{h \in {\{\mathcal{H}(m)\backslash {h^m_0}\}}}  x_{n,h}  \cdot D^{new}_{m}(h^m_0,h) \cdot x_{n,h^m_0}(\tau-1) \\
    & = \mathbf{x}^{\mathrm{T}}_{n,m}  \cdot D_m\cdot \mathbf{x}'_{n,m}(\tau-1),
    \end{aligned}
\end{equation}
where vector $\mathbf{x}_{n,m} = [x_{n,h^m_0}, x_{n,h^m_1}, \cdots, x_{n,h^m_{H(m)}}]$, 
and $D_m$ denotes the transit time of model $m$, including the possible switching or addition of the model state during the transit process from the previous window to the current window.

Since DNNs must be cached before providing services, user $u$'s request can be routed to BS $n$ with the cached submodel of $m_u$ for inference only if the model loading completion time is earlier than the request initiation time $s_u$:
\begin{equation}
    \label{con.Tdep}
    \sum_{n\in \mathcal{N}}  y_{n,u} \cdot T^{\text{load}}_{n,m_u} \le s_u, \forall u \in \mathcal{U}.
\end{equation}

\subsection{Problem Formulation}
\q{In the offline scenario where user requests can be predicted in advance, the purpose is to} jointly decide on model caching and request routing variables in a multi-edge collaborative scenario to
maximize the total inference precision of all user requests in each observation window, while the constraints associated with the decision variables are not violated. 
Let $p_h$ be the expected inference precision of the submodel $h$. The problem can be formulated as follows:
\begin{align}
    (\mathcal{P}) \max_{x, y} & \sum_{u\in \mathcal{U}} \sum_{n\in \mathcal{N}}\sum_{h\in \mathcal{H}(m_u)}  x_{n,h}\cdot  y_{n,u}\cdot p_h \label{eq:prob} \\
    \text{s.t.} \quad &(\ref{con.x}), (\ref{con.R}), (\ref{con.y}), (\ref{con.T}), (\ref{con.Tdep}), \nonumber \\
    \quad &  x_{n,h} \in \{0, 1\},  y_{n,u} \in \{0, 1\}, \forall n \in \mathcal{N}, u \in \mathcal{U}, h \in \mathcal{H},
\end{align}
where constraint (\ref{con.x}) ensures that each BS caches one submodel per model type;  
constraint (\ref{con.R}) bounds the memory usage; constraint (\ref{con.y}) restricts each user request to at most one BS; constraint (\ref{con.T}) keeps the end-to-end inference latency within user tolerance; and constraint (\ref{con.Tdep}) requires the DNN models to be cached before executing user requests at BSs.
For ease of reference, important notations are listed in Table~\ref{table:notation}.

\revi{\textbf{Limitations of our model.} In this paper, we focus on the joint optimization of dynamic model caching and request routing. To ensure theoretical tractability, we adopt several simplifying assumptions, including a contention-free communication model and a fault-free operating environment. Moreover, we treat the inference precision $p_h$ as a static expected value, disregarding input-specific variations. While these assumptions facilitate analytical modeling, addressing stochastic and practical factors in real-world environments represents an important direction for future research.}

\section{The CoCaR Approach}
\label{cocar approach}

In this section, we transform problem $\mathcal{P}$, formulated under the assumption in Sec.~\ref{request routing model} that user requests for the upcoming window are known in advance, into a linear programming problem. The assumption allows proactive downloading of the required models from the cloud to the edge server, followed by loading them into the server's memory. Then, we propose CoCaR to solve the 
transformed problem. Furthermore, we provide a detailed theoretical analysis and complexity evaluation for CoCaR, and extend it to practical scenarios.

\subsection{Problem Transformation}

The term $x_{n,h} \cdot y_{n,u}$ in problem $\mathcal{P}$ makes it a nonlinear integer programming problem, which is difficult to solve directly. To address this, we introduce the virtual binary integer variable $A_{n,u,h}$ to transform problem $\mathcal{P}$ into an equivalent integer linear programming problem $\mathcal{P}1$. We define
\begin{equation}
    \label{A}
    A_{n,u,h} = x_{n,h} \cdot y_{n,u}, \forall n \in \mathcal{N}, u \in \mathcal{U}, h \in \mathcal{H}(m_u).
\end{equation}

Then, we introduce three additional constraints to equivalently represent Eq. (\ref{A}):
\begin{equation}
    \begin{aligned}
        \label{con.Axy}
        A_{n,u,h} \le x_{n,h}, A_{n,u,h} \le y_{n,u}, \\
        x_{n,h} + y_{n,u} -1 \le A_{n,u,h}.
    \end{aligned}
\end{equation}

Subsequently, from Eq. (\ref{A}), we can observe that $A_{n,u,h}$ has practical significance, i.e., 
$A_{n,u,h}=1$ indicates that BS $n$ has cached submodel $h$ and user $u$'s request is routed to BS $n$.
Thus, we can eliminate $y$ to reduce the number of variables and simplify the problem.

Therefore, problem $\mathcal{P}1$ involves only the variables $x_{n,h}$ and $A_{n,u,h}$.
For each user $u \in \mathcal{U}$, there is
\begin{equation}
    \sum_{n\in \mathcal{N}}\sum_{h\in \mathcal{H}(m_u)} A_{n,u,h} = \sum_{n\in \mathcal{N}}\sum_{h\in \mathcal{H}(m_u)} x_{n,h} \cdot y_{n,u} = \sum_{n\in \mathcal{N}} y_{n,u},
\end{equation}
which holds due to Eq. (\ref{con.x}). Thus, the constraint ($\ref{con.y}$) can be equivalently transformed into:
\begin{equation}
    \label{con.A}
    \sum_{n\in \mathcal{N}}\sum_{h\in \mathcal{H}(m_u)} A_{n,u,h} \le 1, \forall u \in \mathcal{U}.
\end{equation}

Furthermore, for the $T^{\text{off}}_{u}$ part of the end-to-end inference latency in Eq.~(\ref{con.T}), let $T^{\text{off}}_{u} = \sum_{n\in \mathcal{N}}\sum_{h\in \mathcal{H}(m_u)}x_{n,h} \cdot y_{n,u} \cdot (\frac{d_u}{\varphi_{\hat{n}_u}} + \frac{d_u}{r_{\hat{n}_u,n}} + \lambda_{u,n}), \forall u \in \mathcal{U},$ which is equivalent to the original constraint. 
Consequently, problem $\mathcal{P}1$ can be formulated as:
\begin{align}
    (\mathcal{P}1) \max_{x, A} & \sum_{u\in \mathcal{U}} \sum_{n\in \mathcal{N}}\sum_{h\in \mathcal{H}(m_u)} A_{n,u,h}\cdot p_h \label{eq:prob2} \\
    \text{s.t.} \quad &(\ref{con.x}), (\ref{con.R}), (\ref{con.A}), \nonumber\\
    \quad & A_{n,u,h} \le x_{n,h}, \forall n \in \mathcal{N}, u \in \mathcal{U}, h \in \mathcal{H}(m_u), \label{con.A<x} \\
    \quad &\sum_{n\in \mathcal{N}}\sum_{h\in \mathcal{H}(m_u)}A_{n,u,h} \cdot ((\frac{d_u}{\varphi_{\hat{n}_u}} + \frac{d_u}{r_{\hat{n}_u,n}} + \lambda_{u,n}) \nonumber\\
    \quad &+ \frac{c_h\cdot d_u}{C_n}) \le ddl_u, \forall u \in \mathcal{U}, \label{con.Addl} \\
    \quad &\sum_{n\in \mathcal{N}}\sum_{h, h'\in \mathcal{H}(m_u)} A_{n,u,h} \cdot x_{n,h'}(\tau-1) \cdot D_{m_u}(h', h) \nonumber\\
    \quad &\le s_u, \forall u \in \mathcal{U}, \label{con.Asu} \\
    \quad & x_{n,h} \in \{0, 1\}, \forall n \in \mathcal{N}, h \in \mathcal{H}, \\
    \quad & A_{n,u,h} \in \{0, 1\},\forall n \in \mathcal{N}, u \in \mathcal{U}, h \in \mathcal{H}(m_u). 
\end{align}

However, problem $\mathcal{P}1$ is still difficult to solve in polynomial time due to its non-convexity. To address this, we relax the variables in $\mathcal{P}1$ and transform it into an LP problem:
\begin{align}
    (\mathcal{P}1\text{-LR}) \max_{x, A} &\sum_{u\in \mathcal{U}} \sum_{n\in \mathcal{N}}\sum_{h\in \mathcal{H}(m_u)} A_{n,u,h}\cdot p_h \label{eq:prob2-lr} \\
    \text{s.t.} \quad &(\ref{con.x}), (\ref{con.R}), (\ref{con.A}), (\ref{con.A<x}), (\ref{con.Addl}), (\ref{con.Asu}), \nonumber \\
    \quad & x_{n,h} \in [0, 1], \forall n \in \mathcal{N}, h \in \mathcal{H}, \\
    \quad & A_{n,u,h} \in [0, 1],\forall n \in \mathcal{N}, u \in \mathcal{U}, h \in \mathcal{H}(m_u). 
\end{align}

\begin{algorithm}[t]
    \caption{CoCaR Algorithm}
    \label{alg:1}
    \small
    \begin{algorithmic}[1]
        \REQUIRE $\{m_u, d_u, ddl_u, s_u\}, \tau, \tilde{x}_{n,h}(\tau-1)$, $\mathcal{N, M}$
        \ENSURE $\tilde{x}_{n,h}, \tilde{A}_{n,u,h}$, $\tilde{y}_{n,u}$
        \STATE Solve problem $\mathcal{P}1\text{-LR}$ and get the optimal fractional solution $x^{\dagger}_{n,h}, A^{\dagger}_{n,u,h}$. \label{alg:1:1}
        \FOR{$n \in \mathcal{N}, m \in \mathcal{M}$} \label{alg:1:for1}
            \STATE Sample $\tilde{\mathbf{x}}_{n,m}$ from the multinoulli distribution $\mathrm{Pr}[\tilde{\mathbf{x}}_{n,m} = \mathbb{I}(h)] = x^{\dagger}_{n,h},\ h\in \mathcal{H}(m)$. \label{alg:1:2}
            \STATE Denote the sample result by $\tilde{\mathbf{x}}_{n,m}=\mathbb{I}(\hat{h})$. \label{alg:1:3}
            \STATE Set $\tilde{x}_{n,\hat{h}} = 1,\ \tilde{x}_{n,h} = 0,\ h\in \mathcal{H}(m)\backslash\{\hat{h}\}$. \label{alg:1:4}
        \ENDFOR
        \FOR{$u \in \mathcal{U}$, $n \in \mathcal{N}$} \label{alg:1:for2}
            \FOR{$h \in \mathcal{H}(m_u)$} \label{alg:1:for3}
                \STATE Set $\tilde{\phi}_{n,u,h} = 1$ with probability $\frac{A^{\dagger}_{n,u,h}}{x^{\dagger}_{n,h}}$. \label{alg:1:5}
                \STATE Let $\tilde{A}_{n,u,h} = \tilde{x}_{n,h} \cdot \tilde{\phi}_{n,u,h}.$ \label{alg:1:6}
            \ENDFOR
            \STATE Let $\tilde{y}_{n,u}=\mathbf{1}(\sum_{h\in H(m_u)}\tilde{A}_{n,u,h}>0)$.\label{alg:1:7}
        \ENDFOR
        \RETURN $\tilde{x}_{n,h}, \tilde{A}_{n,u,h}$, $\tilde{y}_{n,u}$
    \end{algorithmic}
\end{algorithm}

\subsection{Approximation Algorithm}

We propose the CoCaR algorithm to solve the JDCR problem using random rounding and \q{linear programming}. The details are provided below and summarized in Alg. \ref{alg:1}.

First, by applying a linear programming solver \cite{lee2015} to solve problem $\mathcal{P}1\text{-LR}$, we can obtain the optimal fractional solutions $x^{\dagger}_{n,h}$ and $A^{\dagger}_{n,u,h}$ (Line \ref{alg:1:1}).
Then, we round $x^{\dagger}_{n,h}$ and $A^{\dagger}_{n,u,h}$ to obtain vector $\tilde{\mathbf{x}}_{n,m}$ and intermediate variable $\tilde{\phi}_{n,u,h}$,
where $\tilde{\mathbf{x}}_{n,m}$ contains the integer solutions $\tilde{x}_{n,h}$,
and $\tilde{\phi}_{n,u,h}$ indicates whether user $u$'s request attempts to route to submodel $h$ at BS $n$ for inference, serving to determine the integer solution $\tilde{A}_{n,u,h}$. Specifically, let the vector $\tilde{\mathbf{x}}_{n,m}$ be the indicator vector $\mathbb{I}(h)$ with probability $x^{\dagger}_{n,h}$ (Lines \ref{alg:1:2}-\ref{alg:1:4})
and the rounding probability of $\tilde{\phi}_{n,u,h}$ be $\phi^{\dagger}_{n,u,h}$ (Lines \ref{alg:1:5}-\ref{alg:1:6}):
\begin{equation}
    \begin{aligned}
    \label{P[x]}
    &\mathrm{Pr}[\tilde{\mathbf{x}}_{n,m} = \mathbb{I}(h)] = x^{\dagger}_{n,h}, \\
    &\mathrm{Pr}[\tilde{\phi}_{n,u,h} = 1] = \phi^{\dagger}_{n,u,h} = \frac{A^{\dagger}_{n,u,h}}{x^{\dagger}_{n,h}},
    \end{aligned}
\end{equation}
where $\mathbb{I}(h)$ is a one-hot vector of size $|\mathcal{H}(m)|$, with the $h$-th element being 1 and 0 otherwise, indicating that the submodel $h$ of model $m$ is cached at BS $n$, i.e., $\tilde{x}_{n,h}=1$. Finally, let
\begin{equation}
    \label{P[A_nuh]}
    \tilde{A}_{n,u,h} = \tilde{x}_{n,h} \cdot \tilde{\phi}_{n,u,h}.
\end{equation}
Thus, we can revert $\tilde{y}_{n,u}$ based on $\tilde{A}_{n,u,h}$ (Line {\ref{alg:1:7}}).

Then, we provide theoretical guarantees on the quality of the solutions returned by the CoCaR algorithm. 

\begin{lemma}
The solutions returned by CoCaR satisfy the constraints in problem $\mathcal{P}1$ in expectation.    
\end{lemma}

\begin{proof}
First, the inequality $\tilde{A}_{n,u,h} \le \tilde{x}_{n,h}$ holds according to Eq. (\ref{P[A_nuh]}), which satisfies constraint (\ref{con.A<x}).

Second, for integer solutions $\tilde{x}_{n,h}$, $\sum_{h\in \mathcal{H}(m)} \tilde{x}_{n,h} = \|\tilde{\mathbf{x}}_{n,m}\|_1 = 1$ holds.
Therefore, constraint (\ref{con.x}) can be satisfied, where $\|\tilde{\mathbf{x}}_{n,m}\|_1$ is the L1 norm of vector $\tilde{\mathbf{x}}_{n,m}$.

Next, the expected amount of memory consumed by model caching at BS $n, n\in \mathcal{N}$ is given by:
$\mathbb{E}[\sum_{h\in \mathcal{H}} \tilde{x}_{n,h} \cdot r_h] 
= \mathbb{E}[\sum_{m\in \mathcal{M}} \tilde{\mathbf{x}}_{n,m} \cdot \mathbf{r}_m] 
= \sum_{m\in \mathcal{M}}\sum_{h\in \mathcal{H}(m)} x^{\dagger}_{n,h} \cdot r_h \le R_n,$     
where vector $\mathbf{r}_m = [r_{h^m_0}, r_{h^m_1},\cdots, r_{h^m_{H(m)}}]$, and the last inequality holds due to constraint (\ref{con.R}).

Then, constraint (\ref{con.A}) is satisfied in expectation due to:
\begin{align}
    &\mathbb{E}[\sum_{n\in \mathcal{N}} \sum_{h\in \mathcal{H}(m_u)} \tilde{A}_{n,u,h}] 
    = \sum_{n\in \mathcal{N}} \mathbb{E}[\tilde{\mathbf{x}}_{n,m_u} \cdot \vec{\tilde{{\phi}}}_{n,u}] &\notag \\
    &= \sum_{n\in \mathcal{N}} \sum_{h\in \mathcal{H}(m_u)} \mathrm{Pr}[\tilde{\mathbf{x}}_{n,m_u} = \mathbb{I}(h)] \cdot \mathrm{Pr}[\tilde{\phi}_{n,u,h} = 1] \\
    &= \sum_{n\in \mathcal{N}} \sum_{h\in \mathcal{H}(m_u)} x^{\dagger}_{n,h} \cdot \frac{A^{\dagger}_{n,u,h}}{x^{\dagger}_{n,h}} 
    = \sum_{n\in \mathcal{N}} \sum_{h\in \mathcal{H}(m_u)} A^{\dagger}_{n,u,h} \le 1,  &\notag
\end{align}
where vector $\vec{\tilde{{\phi}}}_{n,u} = [\tilde{{\phi}}_{n,u,h^{m_u}_0}, \tilde{{\phi}}_{n,u,h^{m_u}_1}, \cdots, \tilde{{\phi}}_{n,u,h^{m_u}_{H(m_u)}}]$. 
The first equation holds due to vector operations,
and the third due to Eq. (\ref{P[x]}).
The inequality holds by constraint (\ref{con.A}).

Furthermore, the end-to-end inference latency perceived by user $u\in \mathcal{U}$ is expected to be:
\begin{align}
    &\mathbb{E}[\sum_{n\in \mathcal{N}} \sum_{h\in \mathcal{H}(m_u)} \tilde{A}_{n,u,h} \cdot ((\frac{d_u}{\varphi_{\hat{n}_u}} + \frac{d_u}{r_{\hat{n}_u,n}} + \lambda_{u,n}) + \frac{c_h\cdot d_u}{C_n})] & \notag \\
    &= \mathbb{E}[\sum_{n\in \mathcal{N}} \vec{\tilde{\phi}}_{n,u} \cdot (\tilde{\mathbf{x}}_{n,m_u} \odot \vec{\hat{\mathbb{T}}}_{n,u})] & \label{E[ddl]} \\
    &= \sum_{n\in \mathcal{N}} \sum_{h\in \mathcal{H}(m_u)} A^{\dagger}_{n,u,h} \cdot \hat{\mathbb{T}}_{n,u,h} \le ddl_u, & \notag
\end{align}
where $\hat{\mathbb{T}}_{n,u,h}=((\frac{d_u}{\varphi_{\hat{n}_u}} + \frac{d_u}{r_{\hat{n}_u,n}} + \lambda_{u,n}) + \frac{c_h\cdot d_u}{C_n})$,
vector $\vec{\hat{\mathbb{T}}}_{n,u} = [\hat{\mathbb{T}}_{n,u,h^{m_u}_0}, \hat{\mathbb{T}}_{n,u,h^{m_u}_1}, \cdots, \hat{\mathbb{T}}_{n,u,h^{m_u}_{H(m_u)}}]$ and $\odot$ represents the element-wise vector multiplication.
The second equation holds due to Eq. (\ref{P[x]}), and the inequality is by constraint (\ref{con.Addl}).

Finally, the model loading latency for user $u, u\in \mathcal{U}$ is expected to be:
\begin{align}
    &\mathbb{E}[\sum_{n\in \mathcal{N}}\sum_{h\in \mathcal{H}(m_u)}\sum_{h'\in \mathcal{H}(m_u)} \tilde{A}_{n,u,h} \cdot x_{n,h'}(\tau-1) \cdot D_{m_u}(h', h)] & \notag\\
    &= \mathbb{E}[\sum_{n\in \mathcal{N}} \vec{\tilde{\phi}}_{n,u} \cdot (\tilde{\mathbf{x}}_{n,m_u} \odot \vec{\hat{D}}_{n,m_u})] \\
    &= \sum_{n\in \mathcal{N}}\sum_{h, h'\in \mathcal{H}(m_u)} A^{\dagger}_{n,u,h} \cdot x_{n,h'}(\tau-1) \cdot \hat{D}_{n,h} \le s_u, & \notag
\end{align}
where $\hat{D}_{n,h} = \sum_{h'\in \mathcal{H}(m_u)} x_{n,h'}(\tau-1) \cdot D_{m_u}(h', h)$, and the vector $\vec{\hat{D}}_{n,m_u} = [\hat{D}_{n,h^{m_u}_0}, \hat{D}_{n,h^{m_u}_1}, \cdots, \hat{D}_{n,h^{m_u}_{H(m_u)}}]$. The last inequality holds due to constraint (\ref{con.Asu}). 
\end{proof}

\begin{lemma}
The objective value obtained by the CoCaR is equal to that of the optimal fractional solution in expectation.    
\end{lemma}

\begin{proof}
First, let $\mathbb{P}^{\dagger}$ be the objective value in the optimal fractional solution, $\mathbb{P}^{\dagger}=\sum_{u\in \mathcal{U}} \sum_{n\in \mathcal{N}} \sum_{h\in \mathcal{H}(m_u)} A^{\dagger}_{n,u,h} \cdot p_h$.
In expectation, the total inference precision of user requests that the CoCaR algorithm can obtain is:
\begin{equation}
    \begin{aligned}
        \label{E[P]}
        \mathbb{E}[\tilde{\mathbb{P}}] 
        &= \mathbb{E}[ \sum_{u\in \mathcal{U}} \sum_{n\in \mathcal{N}}\sum_{h\in \mathcal{H}(m_u)} \tilde{A}_{n,u,h}\cdot p_h] \\
        &=  \sum_{u\in \mathcal{U}} \sum_{n\in \mathcal{N}} \mathbb{E}[\vec{\tilde{\phi}}_{n,u} \cdot (\tilde{\mathbf{x}}_{n,m_u} \odot \vec{P}_{m_u})]\\
        &=  \sum_{u\in \mathcal{U}} \sum_{n\in \mathcal{N}} \sum_{h\in \mathcal{H}(m_u)} A^{\dagger}_{n,u,h} \cdot p_h = \mathbb{P}^{\dagger},
    \end{aligned}
\end{equation}
where vector $\vec{P}_{m_u} = [p^{m_u}_0, p^{m_u}_1, \cdots, p^{m_u}_{H(m_u)}]$. 
\end{proof}

The above lemmas show that CoCaR has theoretical guarantees in expectation.
However, in practice, random rounding may lead to constraint violations. 
Next, we analyze in detail using the Chernoff Bound theorem \cite{chernoff}.

\begin{definition}[\bf Chernoff Bound \cite{chernoff}]
Given $I$ independent random variables $z_1, z_2, \dots, z_I$, where for all $z_i$ obey multinoulli distribution, and $z_i \in [0, 1]$. 
Let $\mu = \mathbb{E}[\sum_{i=1}^I z_i]$. Then, it holds that 
$\mathrm{Pr}[\sum_{i=1}^I z_i \ge (1+\delta)\mu] \le exp^{\frac{-\delta^2 \mu}{2+\delta}},\  \forall\ \delta>0$ and
$\mathrm{Pr}[\sum_{i=1}^I z_i \le (1-\delta)\mu] \le exp^{\frac{-\delta^2 \mu}{2}},\ \forall\ 0<\delta<1$.
\end{definition}

\begin{theorem}
There is a high probability that the solution returned by the CoCaR algorithm has at least 
$(1-\sqrt{\frac{4\ln|\mathcal{H}|}{\mathbb{P}^{\dagger}}})^2$ 
approximation ratio with the optimal integer solution,
where $\mathbb{P}^{\dagger}$ is the objective value obtained from the optimal fractional solution,
under the assumption $\mathbb{P}^{\dagger} \ge 4\ln|\mathcal{H}|$.
\end{theorem}

\begin{proof}
According to Eq.~(\ref{P[x]}), $\mathbb{P}^{\dagger}$ is equivalent to
$\mathbb{P}^{\dagger} =  \sum_{u\in \mathcal{U}} \sum_{n\in \mathcal{N}} \sum_{h\in \mathcal{H}(m_u)} x^{\dagger}_{n,h} \cdot \phi^{\dagger}_{n,u,h} \cdot p_h$.
Next, to satisfy the independence requirement for the Chernoff Bound, we provide a two-stage proof.

In stage I, we first keep $\phi^{\dagger}_{n,u,h}$ in $\mathbb{P}^{\dagger}$ unchanged and round $x^{\dagger}_{n,h}$ to $\tilde{x}_{n,h}$ to obtain $\mathbb{P}'$:
\begin{equation}
    \begin{aligned}
        \mathbb{P}'&=  \sum_{u\in \mathcal{U}} \sum_{n\in \mathcal{N}} \sum_{h\in \mathcal{H}(m_u)} \tilde{x}_{n,h} \cdot \phi^{\dagger}_{n,u,h} \cdot p_h \\
        &=  \sum_{n\in \mathcal{N}} \sum_{m\in \mathcal{M}} (\tilde{\mathbf{x}}_{n,m} \cdot \sum_{u\in \{u|m_u=m\}} \vec{P}_u), 
    \end{aligned}
\end{equation}
where $\vec{P}_u = [\frac{A^{\dagger}_{n,u,h^{m_u}_0}}{x^{\dagger}_{n,h^{m_u}_0}} \cdot p_{h^{m_u}_0}, \cdots, \frac{A^{\dagger}_{n,u,h^{m_u}_{H(m_u)}}}{x^{\dagger}_{n,h^{m_u}_{H(m_u)}}}\cdot p_{h^{m_u}_{H(m_u)}}]$.

Denote $v_{n,m} = \tilde{\mathbf{x}}_{n,m} \cdot\sum_{u\in \{u|m_u=m\}} \vec{P}_u$, 
the determinism of $\sum_{u\in \{u|m_u=m\}} \vec{P}_u$ and the independence of $\tilde{\mathbf{x}}_{n,m}$ ensure that $v_{n,m}$ are independent.
Clearly, we have $\mathbb{E}[\mathbb{P}'] = \mathbb{P}^{\dagger}$, and by the Chernoff Bound theorem, when $\delta_1 \in [0, 1]$, we have
\begin{equation}
    \label{q1}
    \mathrm{Pr}[\mathbb{P}' \le (1-\delta_1)\mathbb{P}^{\dagger}] = q_1 \le exp(\frac{-\delta_1^2 \mathbb{P}^{\dagger}}{2}).
\end{equation}

Next, find a $\delta_1$ value for the right side of Eq. (\ref{q1}) to make it very small. Specifically, we require,
\begin{equation}
    \label{delta1}
    exp(\frac{-\delta_1^2 \mathbb{P}^{\dagger}}{2}) \le \frac{1}{|\mathcal{H}|^2}.
\end{equation}
This means that as the number of submodels increases, the probability bound converges to zero quickly.
Eq. (\ref{delta1}) holds when $\delta_1 \ge \sqrt{\frac{4\ln|\mathcal{H}|}{\mathbb{P}^{\dagger}}}$ 
, which is true if we pick
$\delta_1 = \sqrt{\frac{4\ln|\mathcal{H}|}{\mathbb{P}^{\dagger}}}$.

In stage II, building upon the rounded result $\tilde{x}_{n,h}$ in stage I, we round $\phi^{\dagger}_{n,u,h}$ to $\tilde{\phi}_{n,u,h}$, resulting in the equation $\tilde{\mathbb{P}}$:  
$ \tilde{\mathbb{P}} =  \sum_{u\in \mathcal{U}} \sum_{n\in \mathcal{N}} \sum_{h\in \mathcal{H}(m_u)} \tilde{x}_{n,h} \cdot \tilde{\phi}_{n,u,h} \cdot p_h$.
Since $\tilde{x}_{n,h} \cdot p_h$ are deterministic and $\tilde{\phi}_{n,u,h}$ are independent,
it follows that: $\mathbb{E}[\tilde{\mathbb{P}}] =  \sum_{u\in \mathcal{U}} \sum_{n\in \mathcal{N}} \sum_{h\in \mathcal{H}(m_u)} \tilde{x}_{n,h} \cdot \mathrm{Pr}[\tilde{\phi}_{n,u,h}=1] \cdot p_h
= \mathbb{P}'$.
By the Chernoff Bound theorem, when $\delta_2 \in [0, 1]$, we have:
\begin{equation}
    \label{c}
    \mathrm{Pr}[\tilde{\mathbb{P}} \le (1-\delta_2)\mathbb{P}'] = q_2 \le exp(\frac{-\delta_2^2 \mathbb{P}'}{2}).
\end{equation}

Since stage I and stage II are independent, when the condition $\mathbb{P}' \ge (1-\delta_1)\mathbb{P}^{\dagger}$ in Eq. (\ref{q1}) is satisfied, we have:
\begin{equation}
    \begin{aligned}
        \label{q2}
        q_2 &= \mathrm{Pr}[\tilde{\mathbb{P}} \le (1-\delta_2)\mathbb{P}' | \mathbb{P}' \ge (1-\delta_1)\mathbb{P}^{\dagger}]\\
        &\le exp(\frac{-\delta_2^2 \mathbb{P}'}{2}) \le exp(\frac{-\delta_2^2 (1-\delta_1) \mathbb{P}^{\dagger}}{2}).
    \end{aligned}
\end{equation}

Next, a value of $\delta_2$ is specified to make the right side of Eq. (\ref{q2}) become very small.  Specifically, we require:
\begin{equation}
    \label{delta2}
    exp(\frac{-\delta_2^2 (1-\delta_1) \mathbb{P}^{\dagger}}{2}) \le \frac{1}{|\mathcal{H}|^{2(1-\delta_1)}}.
\end{equation}
Eq. (\ref{delta2}) holds when $\delta_2$ satisfies $\delta_2 \ge \sqrt{\frac{4\ln|\mathcal{H}|}{\mathbb{P}^{\dagger}}}$. We select $\delta_2 = \sqrt{\frac{4\ln|\mathcal{H}|}{\mathbb{P}^{\dagger}}}$ to make the inequality hold.

Under the conditions specified in (\ref{delta1}) and (\ref{delta2}), by combining (\ref{q1}) and (\ref{c}), we obtain:
\begin{align}
    &\mathrm{Pr}[\tilde{\mathbb{P}} \ge (1-\delta_1)(1-\delta_2)\mathbb{P}^{\dagger}] \notag \\
    &\ge  \mathrm{Pr}[\mathbb{P}' \ge (1-\delta_1)\mathbb{P}^{\dagger}] \mathrm{Pr}[\tilde{\mathbb{P}} \ge (1-\delta_2)\mathbb{P}' | \mathbb{P}' \ge (1-\delta_1)\mathbb{P}^{\dagger}] \notag\\
    &= (1-q_1)(1-q_2) 
    \ge (1-\frac{1}{|\mathcal{H}|^2})(1-\frac{1}{|\mathcal{H}|^{2(1-\delta_1)}}). 
\end{align}
Since $\tilde{\mathbb{P}}^* \le \mathbb{P}^{\dagger}$, where $\tilde{\mathbb{P}}^*$ is the optimal integer solution, it follows that:
$\mathrm{Pr}[\tilde{\mathbb{P}} \ge (1-\delta_1)(1-\delta_2)\tilde{\mathbb{P}}^*] 
\ge \mathrm{Pr}[\tilde{\mathbb{P}} \ge (1-\delta_1)(1-\delta_2)\mathbb{P}^{\dagger}] 
\ge (1-\frac{1}{|\mathcal{H}|^2})(1-\frac{1}{|\mathcal{H}|^{2(1-\delta_1)}})$. 

In practice, as the number of user requests, BSs, and models increases, \(\mathbb{P}^{\dagger}\) related to \(u \in \mathcal{U}, n \in \mathcal{N}, h \in \mathcal{H}\) will significantly exceed \(4\ln{|\mathcal{H}|}\) (\(\mathbb{P}^{\dagger} \gg 4\ln|\mathcal{H}|\)). Therefore, the objective value obtained by the CoCaR algorithm has an approximation ratio of at least $(1-\delta_1)(1-\delta_2) = (1-\sqrt{\frac{4\ln|\mathcal{H}|}{\mathbb{P}^{\dagger}}})^2$ to the optimal integer solution $\tilde{\mathbb{P}}^*$ with a high probability. 
\end{proof}

\begin{theorem}
For BS $n \in \mathcal{N}$, the memory consumption of the cached model returned by the CoCaR algorithm has a high probability of not exceeding its memory capacity by more than a factor of $(\sqrt{\frac{2\ln|\mathcal{H}|}{\zeta^{\dagger}_n}} + \frac{1}{\sqrt{2}})^2 + \frac{1}{2}$, where $\zeta^{\dagger}_n$ is the memory consumption of BS $n$ in the optimal fractional solution, under the assumption $\zeta^{\dagger}_n \ge \ln|\mathcal{H}|$.
\end{theorem}

\begin{proof}
According to Lemma 1,
$\mathbb{E}[\sum_{h\in \mathcal{H}} \tilde{x}_{n,h} \cdot r_h] 
= \mathbb{E}[\sum_{m\in \mathcal{M}} \tilde{\mathbf{x}}_{n,m} \cdot \mathbf{r}_m]
= \sum_{m\in \mathcal{M}}\sum_{h\in \mathcal{H}(m)} x^{\dagger}_{n,h} \cdot r_h$,
where $\tilde{\mathbf{x}}_{n,m} \cdot \mathbf{r}_m$
are independent random variables
and can be normalized into values within [0,1].
Thus, we apply the Chernoff Bound theorem to show that for any $\delta > 0$:
$\mathrm{Pr}[\sum_{h\in \mathcal{H}} \tilde{x}_{n,h} \cdot r_h \ge (1+\delta) \sum_{m\in \mathcal{M}}\sum_{h\in \mathcal{H}(m)} x^{\dagger}_{n,h} \cdot r_h]
\le e^{\frac{-\delta^2 \sum_{m\in \mathcal{M}}\sum_{h\in \mathcal{H}(m)} x^{\dagger}_{n,h} \cdot r_h}{2+\delta}}$.

Denote $\zeta^{\dagger}_n = \sum_{m\in \mathcal{M}}\sum_{h\in \mathcal{H}(m)} x^{\dagger}_{n,h} \cdot r_h$. Since $\zeta^{\dagger}_n \le R_n$,  it follows that:
$\mathrm{Pr}[\sum_{h\in \mathcal{H}} \tilde{x}_{n,h} \cdot r_h \ge (1+\delta) R_n] \le e^{\frac{-\delta^2 \zeta^{\dagger}_n}{2+\delta}}$.
Then, similar to the proof in Theorem 1, to make the right side of the inequality small, i.e., $e^{\frac{-\delta^2 \zeta^{\dagger}_n}{2+\delta}} \le \frac{1}{|\mathcal{H}|^2}$, $\delta$ must satisfy $\delta \ge \frac{\ln|\mathcal{H}|}{\zeta^{\dagger}_n} + \sqrt{\frac{\ln^2 |\mathcal{H}|}{\zeta^{\dagger 2}_n} + \frac{4\ln|\mathcal{H}|}{\zeta^{\dagger}_n}}$.

We select $\delta = \frac{2\ln|\mathcal{H}|}{\zeta^{\dagger}_n} + 2\sqrt{\frac{\ln|\mathcal{H}|}{\zeta^{\dagger}_n}}$. Since ${\zeta^{\dagger}_n} \ge \ln|\mathcal{H}|$ holds in practice, with a high probability, the memory capacity of BS $n$ will not exceed by more than a factor of $1+\delta = (\sqrt{\frac{2\ln|\mathcal{H}|}{\zeta^{\dagger}_n}} + \frac{1}{\sqrt{2}})^2 + \frac{1}{2}$. 
\end{proof}

By similar proofs, we can prove the following three theorems about constraints (\ref{con.A}), (\ref{con.Addl}), and (\ref{con.Asu}).

\begin{theorem}
For user $u\in \mathcal{U}$, there is a high probability that $\sum_{n\in \mathcal{N}} \sum_{h\in \mathcal{H}(m_u)} \tilde{A}_{n,u,h}$ is no greater than
$(\sqrt{\frac{2\ln|\mathcal{H}|}{\eta^{\dagger}}} + \frac{1}{\sqrt{2}})^2 + \frac{1}{2}$,
where, $\eta^{\dagger} = \sum_{n\in \mathcal{N}} \sum_{h\in \mathcal{H}(m_u)} A^{\dagger}_{n,u,h}$.
\end{theorem}

\begin{theorem}
For user $u\in \mathcal{U}$, CoCaR's end-to-end inference latency
has a high probability of not exceeding the user's tolerable perceived latency
by a factor larger than $(\sqrt{\frac{2\ln|\mathcal{H}|}{\mathbb{T}^{{\dagger}}_u}} + \frac{1}{\sqrt{2}})^2 + \frac{1}{2}$,
where $\mathbb{T}^{\dagger}_u$ is the end-to-end inference latency perceived by user $u$ in the optimal fractional solution.
\end{theorem}

\begin{theorem}
For model type $m_u$ requested by user $u, u \in \mathcal{U}$, there is a high probability that CoCaR's model loading latency will not exceed user $u$'s request initiation time by more than a factor of $(\sqrt{\frac{2\ln|\mathcal{H}|}{T^{{dep\dagger}}_{m_u}}} + \frac{1}{\sqrt{2}})^2 + \frac{1}{2}$, 
where $T^{{dep\dagger}}_{m_u}$ is the model loading latency of model $m_u$ in the optimal fractional solution.
\end{theorem}

\subsection{Complexity Analysis of CoCaR}
First, we analyze the scale of the variables and constraints in problem $\mathcal{P}1\text{-LR}$.
Let $h^* = \max_{m} |\mathcal{H}(m)|, m \in \mathcal{M}$, thus the number of variables $x_{n,h}$ and $A_{n,u,h}$ are bounded by $O(N|\mathcal{H}|)$ and $O(NUh^*)$, respectively.
Therefore, the total number of decision variables is bound by $O(N|\mathcal{H}|) + O(NUh^*)\le O(NMh^*)+O(NUh^*)\le O(NUh^*)$, given that $M < U$ in practice.
Similarly, the number of constraints is bounded by $O(NM) + O(N) + 3O(U) + O(NUh^*)$. Since $M < U < Uh^*$, the total number of constraints is also $O(NUh^*)$.
 
Then, considering the three loops in Lines \ref{alg:1:for1}, \ref{alg:1:for2}, and  \ref{alg:1:for3} of Algorithm \ref{alg:1}, the time complexity of CoCaR is $O(NM) + O(NUh^*) \le O(NUh^*)$, as $M < U$ holds.

\subsection{Extension to Practice}
\label{practice}

Since the CoCaR algorithm's rounded solution may violate a few constraints, we employ a heuristic strategy to convert the rounded solutions \q{$\tilde{x}_{n,h}, \tilde{y}_{n,u} $} into feasible ones $x_{n,h}, y_{n,u} $.  
First, for BS $n\in \mathcal{N}$ that violates the memory capacity constraint, we evaluate the benefit of each model type $m \in \mathcal{M}$ based on user requests and inference precision $p_h$ (where $\tilde{x}_{n,h}=1, h \in \mathcal{H}(m)$). 
We remove the least beneficial submodel and try to cache smaller ones.
If no suitable submodel can be cached, user requests previously routed to that BS are redirected to the cloud. This process continues until the memory constraint (\ref{con.R}) is satisfied.
Next, for users $u \in \mathcal{U}$ violating constraints on maximum tolerable perceived latency or model loading latency, requests are redirected to the cloud until constraints (\ref{con.Addl}) and (\ref{con.Asu}) are met.
Finally, if multiple $y_{n,u} = 1$ for a user $u$, the request is routed to the BS with the highest inference precision. 
These steps ensure all caching and routing decisions satisfy the constraints.

\section{Adaptation to Online Scenarios}
\label{extend approach}

In real-world scenarios, the assumption made in Sec. \ref{request routing model}, i.e., future user requests can be predicted in advance, and the required models can be pre-downloaded from the cloud to edge servers, does not always hold \cite{ton-online1, ton-online2}. Thus, caching decisions can only be made after receiving the current user request, and the user must wait until the required model is cached before being served.
To address this, we further extend the CoCaR approach to an online version, CoCaR-OL, which adjusts the dynamic model caching strategy according to real-time user request patterns.

\begin{figure}[t]
\centering
\includegraphics[width=.485\textwidth]{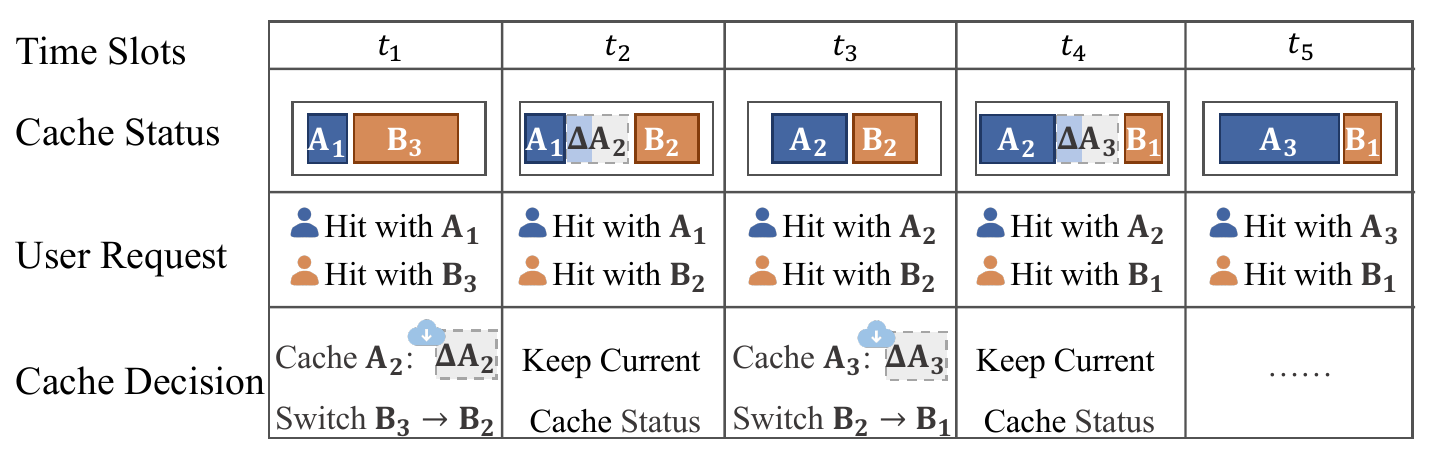}
\caption{
Illustration of model caching in an online scenario. There are two types of models, $A$ and $B$, each comprising three distinct submodels.
At time slot $t_1$, the system decides to cache submodel $A_2$, which requires downloading the additional component $\Delta A_2$ from the cloud to switch from $A_1$ to $A_2$.
To satisfy the cache capacity constraint, submodel $B_3$ is simultaneously reduced to $B_2$.
Since downloading $\Delta A_2$ takes two time slots, $A_2$ becomes available to serve users only from time slot $t_3$ onward. In contrast, model eviction is fast, allowing $B_2$ to serve users immediately at $t_2$.
} 
\label{fig:online_system_model}
\end{figure}

\subsection{System Model}
\label{cocar-ol system model}

As shown in Fig. \ref{fig:online_system_model}, in the online scenario, model caching decisions are made based on historical user request patterns in each time slot.
Therefore, when deciding to cache a larger submodel, the components required for the submodel switching need to be downloaded from the cloud, which may not be completed within the current time slot. 
Since each submodel $h^m_{i+1}$ is typically constructed by incrementally adding components to $h^m_i$, the system can better utilize cached models to serve users by downloading and caching the intermediate submodels in sequence.
For example, if a cached submodel $h^m_1$ is decided to be switched to submodel $h^m_3$, it will first switch to submodel $h^m_2$, and then to submodel $h^m_3$ as the additional required components are downloaded. 
To model this online scenario, we introduce some new notations. Let $W_n$ be the bandwidth between the cloud and base station $n$, and $\Delta t$ be the duration of time slot $t$. Let $\tilde{O}^t_{n,h}$ and $O^t_{n,h}$ denote the sizes of data for submodel $h$ to be downloaded at BS $n$ before and after the caching decision in time slot $t$, respectively. Similarly, $\tilde{X}^t_{n,h}$ and $X^t_{n,h}$ represent the caching status of submodel $h$ at BS $n$ before and after the caching decision in time slot $t$, respectively.

Thus, at time slot $t$, the remaining data size of submodel $h$ to be downloaded at BS $n$ can be calculated as:
\begin{equation}
    \label{eq:tildeO}
    \tilde{O}^t_{n,h} =
    \begin{cases}
        \max(O^{t-1}_{n,h} - \bar{t}\cdot W_n, 0), &\text{if } O^{t-1}_{n,h} \neq 0,\\
        0,  &\text{otherwise},
    \end{cases}
\end{equation}
where $\bar{t} = \Delta t - \min\Big(\frac{\sum_{h' \prec h} O^{t-1}_{n,h'}}{W_n}, \Delta t\Big)$, meaning that each submodel to be downloaded can start downloading immediately after the previous submodel is downloaded.

Then, whether submodel $h$ has finished downloading at time slot $t$ can be denoted by:
\begin{equation}
    \label{eq:g}
    g^t_{n,h} = \mathbf{1}\{O^{t-1}_{n,h} \neq 0 \land \tilde{O}^t_{n,h}=0\},
\end{equation}
where $\mathbf{1}\{a\} = 1$ if condition $a$ is true, and 0 otherwise.

Next, based on the downloading status of submodels, the caching state can be calculated as:
\begin{equation}
    \label{eq:tildeX}
    \tilde{X}^t_{n,h} =
    \begin{cases}
        1,  &  \text{if }  g^t_{n,h} = 1  \land \sum_{h'\succ h} g^t_{n,h'} = 0,\\
        0,  & \text{if } \sum_{h'\succ h} g^t_{n,h'} \ge 1, \\
        X^{t-1}_{n,h}, & \text{otherwise}.
    \end{cases}
\end{equation}
Specifically, if at time slot $t-1$, submodel $h$ at BS $n$ has just been downloaded, and no submodel larger than $h$ has been downloaded in this slot, then submodel $h$ will be cached and $\tilde{X}^t_{n,h}$ is set to 1. 
Conversely, if at time slot $t$, a larger submodel than $h$ that has been downloaded at BS $n$, so that submodel $h$ should not be cached, i.e., $\tilde{X}^t_{n,h}$ is set to 0.
Otherwise, the caching state of submodel $h$ at BS $n$ remains unchanged.

Consequently, based on the above caching state, the inference precision that model type $m$ can provide at BS $n$ during time slot $t$ can be calculated as:
\begin{equation}
    \label{eq:P}
    P_{n,m}(t) = \sum_{h \in \mathcal{H}(m)} \tilde{X}^t_{n,h}\cdot p_h.
\end{equation}

Similar to Eq. (\ref{con.T}), the total end-to-end inference latency at time slot $t$ for routing user requests of model type $m$ from home BS $n'$ to target BS $n$ for inference can be calculated as:
\begin{equation}
    \mathbb{T}^t_m(n',n) = \frac{d_m}{\varphi_{n'}} + \frac{d_m}{r_{n',n}} + \lambda_{n'n} + \sum_{h \in \mathcal{H}(m)} \tilde{X}^t_{n,h} \cdot \frac{c_h\cdot d_m}{C_n},
\end{equation}
where $d_m$ represents the input data size of model type $m$.

Finally, let $Q^t_m(n',n)$ be the Quality of Experience (QoE) at time slot $t$ for routing a user request of model type $m$ with home BS $n'$ to target BS $n$ for inference.
To comprehensively evaluate the trade-off between inference precision and the user-perceived end-to-end inference latency, we define the QoE function as follows:
\begin{equation}
    \label{eq:gain}
    Q^t_m(n',n) = P_{n,m}(t) \cdot \max(0, 1 - (\mathbb{T}^t_m(n',n) - \theta) \cdot \alpha), 
\end{equation}
where $\theta$ is a normalization factor (i.e., the minimum end-to-end inference latency), and $\alpha$ is a smoothing factor that controls the degradation of precision due to end-to-end inference latency. Specifically, when $\mathbb{T}^t_m(n',n) > \theta$, the QoE value $Q^t_m(n',n)$ falls below the original precision $P_{n,m}(t)$, reflecting the negative impact of latency; conversely, if the inference latency is equal to $\theta$, the QoE will not be degraded.

Accordingly, the target BS that maximizes the QoE for a user requesting model $m$ with home BS $n'$ at time slot $t$ can be denoted as:
\begin{equation}
    \label{eq:bestn}
    n^{t*}_m(n') = \arg\max_{n \in \mathcal{N}} Q^t_m(n',n).
\end{equation}

\textbf{Problem Formulation.} In the online model caching and request routing optimization problem, we route user requests to the BS that maximizes their QoE. The optimization objective is to maximize the total QoE of all users by designing an effective caching policy.
The problem is formulated as follows:
\begin{align}
    (\mathcal{P}2)\qquad &\max_{X} \sum_{t\in \mathcal{T}}\sum_{u\in \mathcal{U}} Q^{t}_{m_u}(\hat{n}_u,n^{t*}_m(\hat{n}_u)) \\
    \text{s.t.} \quad & \sum_{h\in \mathcal{H}}X^t_{n,h} \cdot r_h \le R_n, \forall t \in \mathcal{T}, n \in \mathcal{N}\\
    & \mathbb{T}^t_{m_u}(\hat{n}_u,n^{t*}_{m_u}(\hat{n}_u)) \le ddl_u, \forall t \in \mathcal{T}, u \in \mathcal{U}. 
\end{align}

\newcommand{\tO}{{\tilde{O}}}
\newcommand{\tX}{{\tilde{X}}}
\begin{algorithm}[t]
    \caption{CoCaR-OL Algorithm}
    \label{alg:2}
    \small
    \begin{algorithmic}[1]
        \REQUIRE $\mathcal{U}, \mathcal{N}, \mathcal{M}, round$
        \ENSURE $X^{t}_{n,h}, O^{t}_{n,h}$
        \STATE Initialize time slot variable $t=0$.
        \STATE Initialize model caching and download status $X^{0}_{n,h}, O^{0}_{n,h}.$ 
        \WHILE[Start online monitoring and deciding]{True}
            \STATE Step to new time slot: $t\leftarrow t+1$.
            \STATE \textit{\textbf{Routine Update}}: Calculate temporary download status $\tO_{n,h}^t$ according to Eq. \eqref{eq:tildeO}. \label{alg:2:5}
            \STATE \textit{\textbf{Routine Update}}: Calculate download-finishing flag variable $g_{n,h}^t$ and temporary cache status $\tX_{n,h}^t$ according to Eq. \eqref{eq:tildeX}. \label{alg:2:6}
            \STATE Receive current user requests $\{m_u^t\mid u\in \mathcal{U}\}$. \label{alg:2:7}
            \STATE Calculate the best routing destination $n^{t*}_m(n')$ from Eq. \eqref{eq:bestn}.
            \FOR[Route user requests and calc. total QoE]{$u \in \mathcal{U}$}
                \STATE Route use request $m_u$ to the best destination $n_{m_u}^{t^*}(\hat{n}_u)$.
                \STATE Calculate the corresponding QoE value $Q_{m_u}^t$ by Eq. \eqref{eq:gain}.
            \ENDFOR 
            \STATE Update the recent proportion of user requests $f_{n,m}^t$ by Eq.~\eqref{eq:f}.\label{alg:2:13}
            \STATE $O_{n,h}^t := \tO_{n,h}^t, X_{n,h}^t := \tX_{n,h}^t$.
            \FOR[Make model caching decision]{$r = 1 \rightarrow round$} 
                \STATE Randomly choose a BS $n$ for caching adjustment. \label{alg:2:16}
                \STATE Compute the future expected gain $\Delta R$ for each possible cached model  switch based on $O_{n,h}^t, X_{n,h}^t$ according to Eq. \eqref{eq:deltaR}.
                \STATE Find the best scheme for model switching (caching and evicting) by solving a memory-constrained knapsack problem. \label{alg:2:18}
                \STATE \textit{\textbf{Switch Update}}: Update download status $O_{n,h}^t$ by Eq. \eqref{eq:O}.\label{alg:2:19}
                \STATE \textit{\textbf{Switch Update}}: Update cache status $X_{n,h}^t$ by Eq. \eqref{eq:X}. \label{alg:2:20}
            \ENDFOR
        \ENDWHILE
    \end{algorithmic}
\end{algorithm}

\subsection{Algorithm Design}
To solve problem $\mathcal{P}2$, we propose CoCaR-OL, a heuristic caching algorithm based on predictive future gain, which enables fine-grained perception of changes in user request patterns in each time slot and allows flexible adjustment of cached models.
In CoCaR-OL, the action space for enlarging cached submodels is defined as the set of all submodels from the currently cached one up to the first submodel whose additional components, relative to the cached submodels, cannot be fully downloaded within a time slot.

Specifically, we first use user request patterns from the past $\Delta T^P$ time slots to simulate user requests over the next $\Delta T^F$ time slots.
In particular, the proportion of user requests for model type $m$ with home BS $n$ over the past $\Delta T^P$ time slots can be denoted as:
\begin{equation}
    \label{eq:f}
    f^t_{n,m} = \frac{\sum_{i = t-\Delta T^P}^{i = t} \sum_{u\in \mathcal{U}} \mathbf{1}\{m^i_u = m \land \hat{n}^i_u = n\}}{\sum_{i = t-\Delta T^P}^{i = t} U}.
\end{equation}

Next, let $\pi^t_{n,m}(X^t; O^t) \in \{(h',h) \mid X^t_{n,h'} = 1, h \in \mathcal{H}(m)\}$ denote the action taken at time slot $t$ for a cached submodel $h'$ of model type $m$ at BS $n$, under the observation of system states $X^t$ and $O^t$.
When $h' \succ h$, it indicates a switch from the cached submodel $h'$ to a smaller submodel $h$; conversely, it indicates a switch to a larger submodel $h$ or remaining unchanged.
Then, at time slot $t$, the expected future reward of switching a cached submodel $h'$ of model $m$ at BS $n$ to a submodel $h$, while keeping all other system states unchanged, is computed as:
\begin{align}
    R(\pi^t_{n,m} = (h',h)) 
    &= \sum_{t'=t+1}^{\Delta T^F}\sum_{n' \in \mathcal{N}}  
       \gamma^{(t'-t)} \cdot f_{n',m}^{t'}  \notag \\
    &\quad \cdot Q^{t'}_m\!\bigl(n',n^{t'*}_m(n')\bigr),
\end{align}
where $\gamma^t$ denotes the discount factor of the user’s QoE in the $t$th future time slot.

The expected future gain from replacing the cached submodel $h'$ of model $m$ at BS $n$ with $h$ can be computed as:
\begin{align}
\label{eq:deltaR}
    \Delta R(\pi^t_{n,m} = (h',h)) 
    &= R(\pi^t_{n,m} = (h',h)) \notag \\
    &\quad - R(\pi^t_{n,m} = (h',h')).
\end{align}

Finally, after making the caching decision at time slot $t$, if it is decided to switch from the cached submodel $h'$ at BS $n$ to a larger target submodel $\hat{h}^t$, the download status of submodel $h$ can be computed as:
\begin{equation}
    \label{eq:O}
    O^t_{n,h} =
    \begin{cases}
        \Delta r_h,
        \qquad &\text{if } \tilde{O}^t_{n,h} = 0 \land \sum\limits_{h' \prec h} \tilde{X}^{t}_{n,h'} = 1 \land h\preceq \hat{h}^t,  \\
        \tilde{O}^t_{n,h},  &\text{otherwise},
    \end{cases}
\end{equation}
where $\Delta r_h$ represents the additional data size that each submodel $h$ between the cached submodel $h'$ and the target submodel $\hat{h}^t$ needs to download relative to its preceding submodel.
Otherwise, if the decision is to switch from the cached submodel $h'$ at BS $n$ to a smaller target submodel $\hat{h}^t$, the model's cache state can be updated directly as the time required for cache eviction is negligible:
\begin{equation}
    \label{eq:X}
    X^t_{n,h} =
    \begin{cases}
        1,  &  \text{if }  h = \hat{h}^t,\\
        0,  & \text{if } h = h', \\
        \tilde{X}^t_{n,h}, & \text{otherwise}.
    \end{cases}
\end{equation}

The procedure of the CoCaR-OL algorithm is shown in Alg. \ref{alg:2}. At the beginning of each time slot $t$, the download and cache states of each submodel, $\tilde{O}^t_{n,h}$ and $\tilde{X}^t_{n,h}$, are updated through a routine process (Lines \ref{alg:2:5}-\ref{alg:2:6}). Then, the maximum QoE for each user during the current time slot is calculated, and the historical request frequency $f^t_{n,m}$ of each model $m$ at each BS $n$ is updated (Lines \ref{alg:2:7}-\ref{alg:2:13}). In each iteration, a BS $n$ is randomly selected, and all submodels of models that are not being downloaded at BS $n$ are traversed. For each candidate submodel, it is hypothetically enlarged, and a knapsack problem is solved under the constraint of the remaining memory capacity of BS $n$. This step selects a combination of cached submodels from other candidate models and their smaller submodels, such that the expected future gain of cache switching is maximized (Lines \ref{alg:2:16}-\ref{alg:2:18}). The same process is then repeated for other candidate submodels. Finally, the model switching scheme that yields the maximum expected future gain is chosen from all the options, and the corresponding model states $O^t_{n,h}$ and $X^t_{n,h}$ are updated accordingly (Lines \ref{alg:2:19}-\ref{alg:2:20}).

\begin{table}[t]
    \centering
    \caption{Attributes of the Three ViT Submodels}
    \label{vit:att}
    \begin{tabular}{@{}cccccc@{}}
        \toprule
        \textbf{Submodel} & \textbf{Memory (MB)} & \textbf{FLOPs (GFlops)} & \textbf{Precision} \\
        \midrule
        1 & 174.32 & 5.70 & 0.8417 \\
        2 & 227.42 & 7.56 & 0.9413 \\
        3 & 342.05 & 11.29 & 0.9894 \\
        \bottomrule
    \end{tabular}
     \vspace{-2pt}
\end{table}
\begin{table}[t]
    \centering
    \caption{Loading Time of the Three ViT Submodels (s)}
    \label{vit:time}
    \begin{tabular}{@{}c|ccc@{}}
    \toprule
    \multirow{2}{*}[-3pt]{\textbf{\begin{tabular}[c]{@{}c@{}}Original\\ Submodels\end{tabular}}} & \multicolumn{3}{c}{\textbf{Final Submodels}} \\ \cmidrule(l){2-4}
               & \textbf{1} & \textbf{2} & \textbf{3} \\ \midrule
    /          & \makebox[1.5cm]{0.68860} & \makebox[1.5cm]{0.87696} & \makebox[1.5cm]{1.05821}    \\
    \textbf{1} & 0.00000    & 0.24794    & 0.46098    \\
    \textbf{2} & 0.04238    & 0.00000    & 0.25082    \\
    \textbf{3} & 0.04725    & 0.04242    & 0.00000    \\ \bottomrule
    \end{tabular}
\end{table}

\subsection{\revi{Complexity Analysis of CoCaR-OL}}
\revi{The complexity of Algorithm~\ref{alg:2} is primarily determined by system state updates and the caching decision process. First, the complexity of the routine state updates and user routing (Lines 5-14) in each time slot is $O(NH + N^2 + U)$. Second, the primary computational overhead arises from the $\textit{round}$ iterations within the caching decision loop (Lines~15–21). Specifically, Line~17 involves simulating and evaluating variations in the system state and the expected gain over $\Delta T^F$ future time slots, with a complexity of $O(round \cdot \Delta T^F \cdot N H^2)$. Furthermore, solving the caching optimization problem in Line~18 incurs an overhead of $O(\textit{round} \cdot M \cdot H \cdot V)$, where $V$ denotes the number of discrete capacity states in the memory-constrained knapsack problem. Additionally, the complexity of the switch update operation (Lines~19-20) is $O(NH)$. 
Thus, the total complexity of the algorithm is $O(\textit{round} \cdot (\Delta T^F \cdot N H^2 + M \cdot H \cdot V))$.}

\section{Performance Evaluation}
\label{experiment}

In this section, we conduct simulation experiments to evaluate CoCaR after converting its solutions into feasible ones (Sec.~\ref{practice}), as well as CoCaR-OL in online scenarios.

\subsection{Evaluation Setup}
\label{para}

Following a setup similar to \cite{poularakis} and \cite{yao}, the duration of each observation window is set to $\Delta \tau = 3$s and the number of windows $|\Gamma| = 10$, i.e., the total time is 30s. 
We consider an MEC network with $N=5$ BSs and $U=600$ users in each window. The coverage range of a BS is 150\,m, and the connection between BSs is established using the Erdős–Rényi random graph model \cite{er}. 
The wireless uplink transmission rate is set to $\varphi_{\hat{n}_u} = 20$\,Mbps, and the transmission rate between BSs is set to $r_{\hat{n}_u n} = 100$\,Mbps. 
\q{The bandwidth between the cloud and the base station is set to $W_n = 800$\,Mbps.}
The propagation time for one hop is set to 0.01s. We set the memory capacity for each BS $n$ to $R_n = 500$MB and the maximum computing power $C_n = 70$\,Gflops/s.

We use $M = 8$ DNN model types (e.g., ViT, swintransformer \cite{swin}), each with three submodels. Each submodel is trained and tested on the CIFAR-10 dataset \cite{cifar10}. Taking ViT as an example. Tables \ref{vit:att} and \ref{vit:time} show the attributes and loading times of the ViT submodels, respectively. 
Table \ref{vit:att} shows the differences in memory consumption and inference precision among the ViT submodels. Table \ref{vit:time} indicates that loading submodel 2 directly \q{from the secondary storage} takes 0.877s, whereas switching from submodel 1 to submodel 2 takes only 0.248s, thereby increasing the valid service time of ViT. The popularity of model types follows a Zipf distribution with a skewness coefficient of 0.8 \cite{zipf}, and each user generates a request for a random model type.
\q{The data size of each request and model input is $d_u = d_m = 0.144$\,MB.} 
The maximum tolerable latency for the user is set to 0.3s. 
We implemented the CoCaR using Python 3.10.14 and conducted all experiments on an Intel(R) Xeon(R) Gold 6230R CPU @ 2.10\,GHz with eight NVIDIA GeForce RTX 3070 GPUs.

\begin{table}[t]
\color{black}
    \centering
    \caption{Performance Comparison in Offline Scenario (HR: Hit Rate)}
    \label{tab:offline_res}
    \vspace{-6pt}
    \begin{tabular}{@{}lccc@{}}
        \toprule
        \textbf{Method} & \textbf{Avg. Precision} & \textbf{Avg. HR} & \textbf{Avg. Memory Util.} \\
        \midrule
        LR (Upper Bound) & 0.931 & -- & -- \\
        \textbf{CoCaR (Ours)} & \textbf{0.861} & \textbf{0.939} & \textbf{0.866} \\
        GatMARL & 0.460 & 0.521 & 0.818 \\
        Greedy & 0.430 & 0.473 & \textbf{0.867} \\
        $\text{SPR}^3$ & 0.414 & 0.462 & 0.728 \\
        Random & 0.152 & 0.169 & 0.705 \\
        \bottomrule
    \end{tabular}
    \vspace{-2pt}
\end{table}

\subsection{Benchmarks and Evaluation Metrics}
We compare CoCaR with the following algorithms.

\begin{itemize}
    \item \textbf{Linear-Relaxation (LR).} The optimal fractional solution of problem $\mathcal{P}1\text{-LR}$, obtained via a linear programming solver, serves as an upper bound for the optimal integer solution but is typically unattainable in practice.
    \item \textbf{$\text{SPR}^3$} \cite{poularakis}. It uses random rounding to solve the joint service caching and request routing problem under memory, computation, and communication constraints.
    \item \textbf{GatMARL} \cite{gatmarl}. GatMARL models the MEC environment as an undirected graph and applies a graph-attention-based multi-agent reinforcement learning algorithm to learn task offloading and service caching policies.
    \item \textbf{Greedy.} The model type is selected based on popularity, and one of its submodels is cached in descending order of precision. User requests are routed to the home BS.
    \item \textbf{Random.} The models cached at each BS are randomly selected from the submodels associated with each model type. User requests are randomly routed to a BS.
\end{itemize}
Note that $\text{SPR}^3$ and GatMARL only cache the complete service, and all benchmarks except LR ignore the impact of model loading time on decision-making.

Our evaluation is based on the following metrics:

\begin{itemize}
    \item \textbf{Average inference precision.} 
    It reflects the average inference precision for all user requests over all windows.
    \item \q{\textbf{Average QoE of users.} It reflects the average QoE for all user requests under the precision and latency tradeoff over all time slots, which is calculated according to Eq.~(\ref{eq:gain}) in Sec. \ref{cocar-ol system model}.}
    
    \item \textbf{Average hit rate.} 
    It reflects the ratio of requests served by the BS over all windows.
    \item \textbf{Average memory utilization.} 
    It reflects the memory usage ratio of cached models at the BS over all windows.
\end{itemize}

\revi{Table~\ref{tab:offline_res} summarizes the performance in the offline scenario under the settings described in Sec.~\ref{para}. The results demonstrate that CoCaR achieves performance closest to the theoretical optimum (i.e., LR) and significantly outperforms other baselines. As LR is used solely to establish a theoretical upper bound for the optimization objective and does not correspond to a practically deployable model, its hit rate and memory utilization are omitted from the table.}

\begin{figure}[t]
\centering
\includegraphics[width=.48\textwidth]{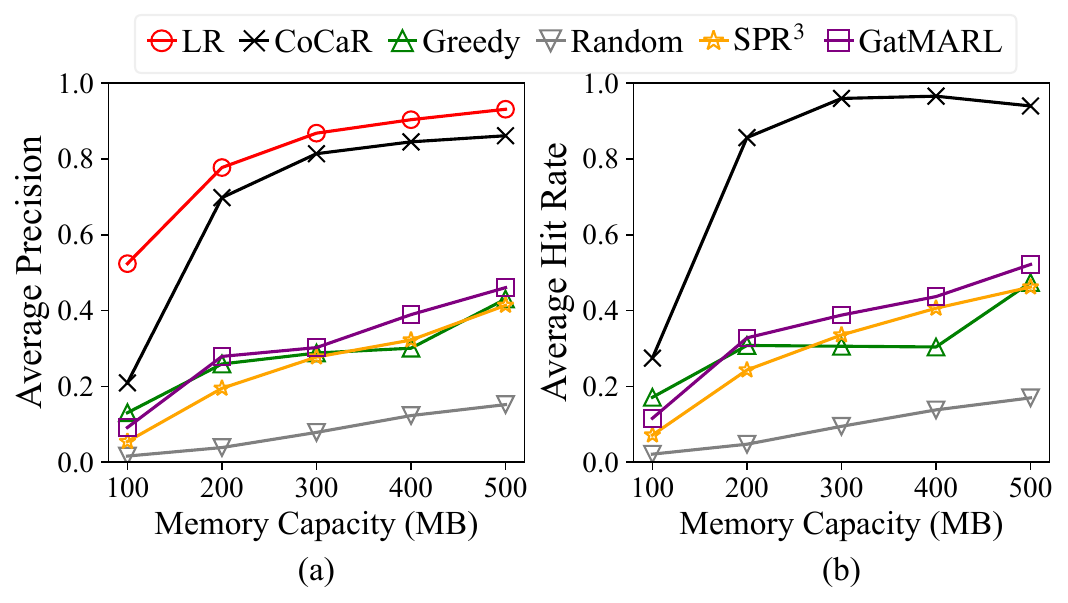}
\vspace{-6pt}
\caption{Impact of different BS memory capacities: (a) Average inference precision; (b) Average hit rate.}
\label{fig:R} 
\end{figure}

\begin{figure}[t]
\centering
\includegraphics[width=.48\textwidth]{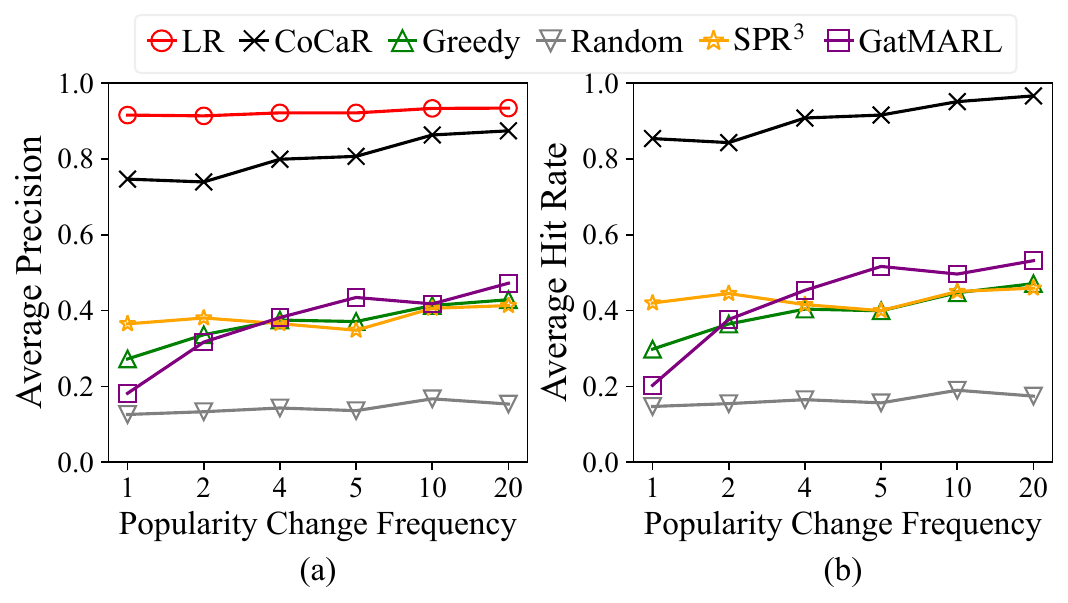}
\vspace{-6pt}
\caption{Impact of popularity change frequency: (a) Average inference precision; (b) Average hit rate. The x-axis indicates how many observation windows the popularity changes once.}
\label{fig:Chg} 
\end{figure}

\begin{figure}[t]
\centering
\includegraphics[width=.46\textwidth]{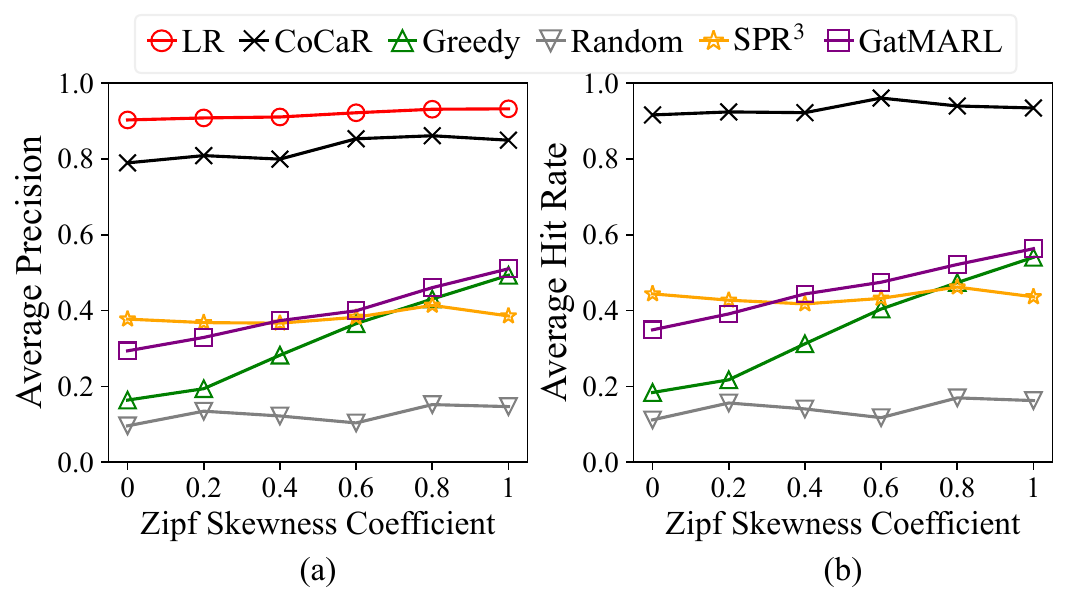}
\vspace{-6pt}
 \caption{Impact of different Zipf skewness coefficients:(a) Average inference precision; (b) Average hit rate.}
\label{fig:Z} 
\end{figure}

\subsection{Evaluation Results and Analysis of CoCaR}

\textbf{Impact of BS Memory Capacity.}
Fig. \ref{fig:R} shows the results of average inference precision and average hit rate as the BS memory capacity varies. As the BS memory capacity increases from 100\,MB to 500\,MB, the performance of all methods improves.  This is because a larger memory capacity allows BSs to cache more models, thereby satisfying more user requests. 
When $R = 500$\,MB, CoCaR's average inference precision gap with LR is only 7.5\%, outperforming other baselines by at least 40.1\%. Therefore, CoCaR, with its superior decision-making mechanism, can quickly adapt to changes in memory resources, consistently outperform other baselines, and steadily converge to the optimal LR algorithm.

\textbf{Impact of Popularity Change Frequency.}
In this experiment, we set the number of observation windows to $\Gamma = 20$ and configured the other parameters to their default values (Sec. \ref{para}). 
Fig. \ref{fig:Chg} shows how often the model popularity changes, i.e., the number of observation windows in which the DNN models' popularity changes once, affecting the algorithms' performance. 
CoCaR consistently achieves the best results after LR, owing to its ability to leverage BS caching results from the previous observation window during decision-making and balance the impact of model loading time by switching among dynamic DNN submodels, thereby quickly caching DNN models and adapting to environmental changes.
In contrast, the other baselines make independent decisions for each observation window, resulting in more fluctuating results in response to popularity changes.
Even when popularity changes in every window, CoCaR's average precision exceeds other baselines by more than 38.2\%, second only to the LR, while maintaining the highest hit rate.

\begin{figure}[t]
\centering
\includegraphics[width=.45\textwidth]{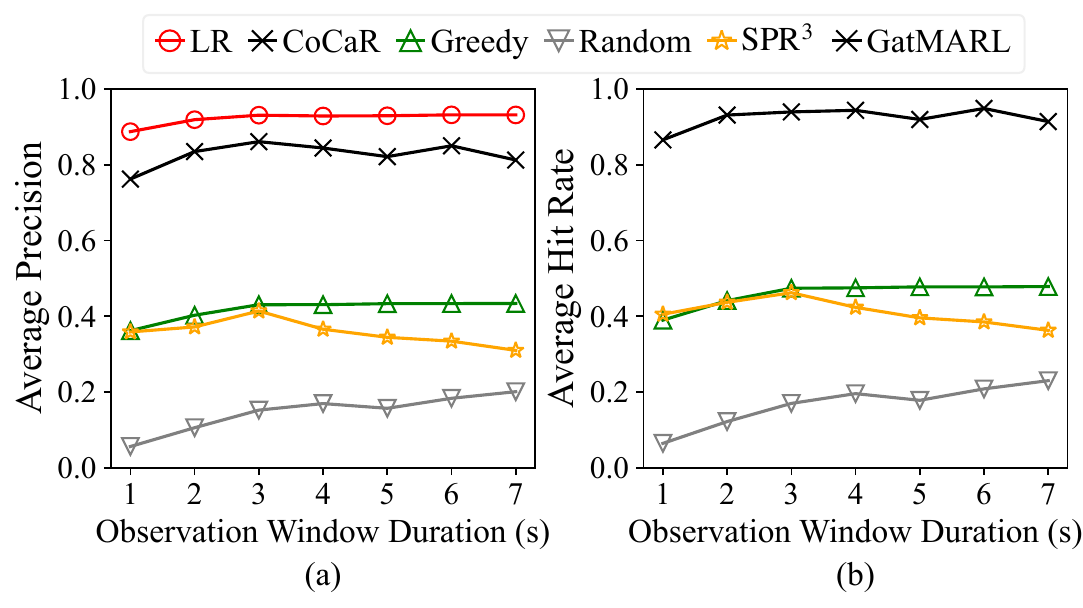}
\vspace{-6pt}
\caption{Impact of different observation window duration: (a) Average inference precision; (b) Average hit rate.}
\label{fig:Duration} 
\end{figure}

\begin{figure}[t]
\centering
\includegraphics[width=.45\textwidth]{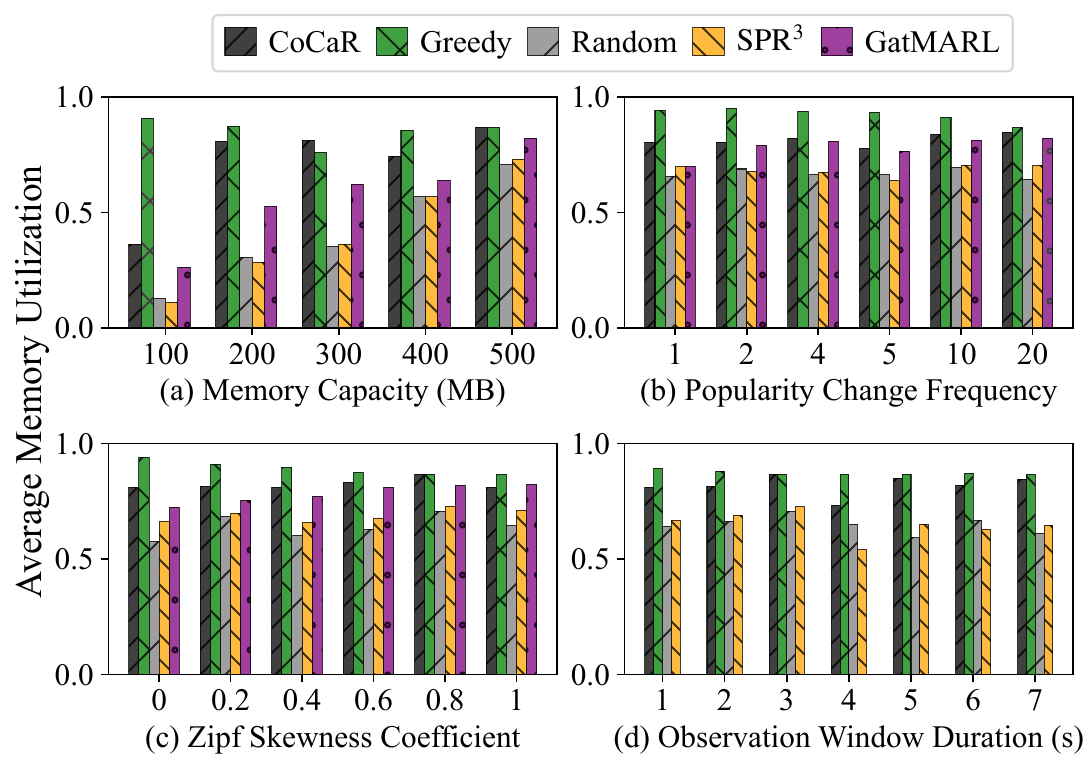}
\vspace{-6pt}
\caption{Average BS memory resource utilization.} 
\label{fig:utilization}
\end{figure}

\textbf{Impact of Zipf Skewness Coefficient.}
Fig. \ref{fig:Z} illustrates the performance of algorithms across Zipf skewness coefficients ranging from 0 to 1.
CoCaR consistently achieves higher average precision and hit rate as Zipf skewness increases, with slight fluctuations due to random rounding. 
Additionally, since a higher coefficient results in more frequent requests for a few highly popular model types, the upward trend in the Greedy algorithm becomes more evident.
When the Zipf skewness coefficient is zero, the popularity vector is uniform, indicating that user requests are completely randomized. Despite this, CoCaR achieves favourable results by leveraging a multi-edge collaboration mechanism and quickly adjusting caching schemes via submodel switching to serve more users.

\textbf{Impact of Observation Window Duration}.
Fig. \ref{fig:Duration} shows the impact of varying the observation window duration on algorithm performance, with the other default settings held constant (Sec. \ref{para}). According to the parameter settings in Sec. \ref{para}, given the observation window duration $\Delta \tau=1$s, there are $\Gamma = 30$ windows, each with $U=200$ users. LR, CoCaR, and $\text{SPR}^3$ exhibit an initial increase followed by a decrease, achieving optimal performance at a duration of 3s.
This is due to the imbalance between model loading time and observation window duration.
A window that is too short significantly increases the time spent on loading models, thereby reducing the valid model service time, inference precision, and hit rate.
Conversely, an overly long window may prevent the algorithm from adapting promptly to changing user requests, thereby affecting performance.
Additionally, fixed hyperparameters, such as the number of users during GarMARL training, limit its decision-making in this experiment.

\textbf{BS Resource Utilization.}
Fig. \ref{fig:utilization} shows the average BS resource utilization of all algorithms under various influencing factors.
\revi{Notably, the Greedy algorithm generally exhibits the highest memory utilization due to its aggressive strategy of prioritizing popular, high-precision submodels, which quickly saturates the limited memory capacity. However, since Greedy relies on independent decision-making at each base station and neglects global information, its average inference precision and hit rate remain significantly lower than CoCaR's. In contrast to merely filling available memory, CoCaR strategically balances submodel size, model-loading latency, and caching diversity from a system-wide perspective.}
Based on the experimental analysis above, CoCaR demonstrates more effective, finer-grained utilization of BS resources, with each unit of resource usage yielding higher average inference precision and hit rate, thereby improving overall service quality and user coverage. Under the experimental setting in Sec. \ref{para}, CoCaR can utilize at least 86\% of BS memory resources.

\begin{figure}[t]
\centering
\includegraphics[width=.48\textwidth]{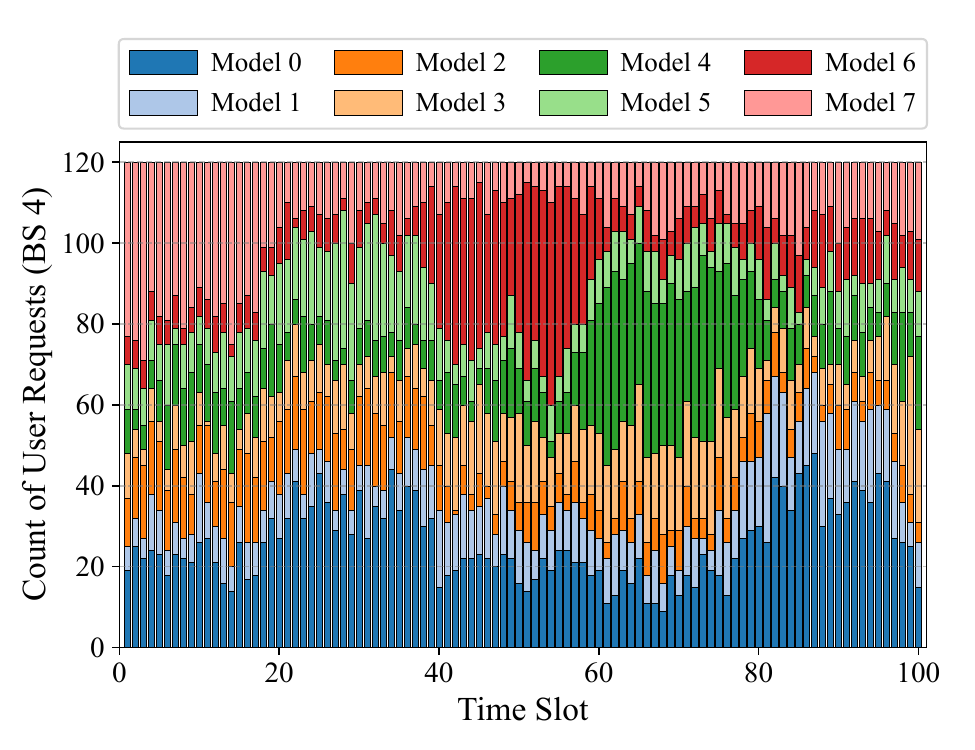}
\vspace{-6pt}
\caption{Count of user requests for DNN models at BS 4 in each time.} 
\label{fig:user_distribution}
\end{figure}

\subsection{Evaluation Results and Analysis of CoCaR-OL}










\textbf{Extended Settings.}
We set the duration of each time slot to $\Delta t=0.5s$, with a total of $\mathcal{T} = 100$ time slots. 
In each time slot, $round = 3$ BSs are randomly selected, and a caching decision is made based on user requests in the past $\Delta T^P = 10$ time slots and the expected future gain in the next $\Delta T^F = 5$ time slots. 
The smoothing factor $\alpha$ and the discount factor $\gamma$ are both set to 0.9. 
In each time slot, $U = 600$ users simultaneously initiate requests for model type $\mathcal{M}$. 
The popularity of DNN models at each BS is updated every 20 time slots, with a warm-up phase starting 5 time slots earlier to gradually adjust it. Other parameters are the same as those defined in Sec. \ref{para}. Fig. \ref{fig:user_distribution} illustrates the distribution of user requests at BS 4 under the above experimental setup.

In the online scenario, we introduce the following online caching algorithms for comparison:

\begin{itemize}

\item \textbf{LFU \cite{lfu}.} 
In each time slot, $round$ BSs are randomly and independently selected. For each selected BS, model request frequencies over the past $\Delta T^P$ time slots at the BS and its one-hop neighbors are computed. 
The cached submodel of the most frequent model is enlarged, while the cached submodel of the least frequent model is gradually reduced until the memory constraint is met. 

\item \textbf{LFU-MAD \cite{lfu-mad}.} It accounts for delayed hits in caching and adjusts the strategy by ranking content. In our experiment, ranks are based on weighted request frequency, with higher weights to more recent user requests. The caching decision is similar to the LFU in this paper.

\item \textbf{Random.} In each timeslot, $round$ BSs are randomly selected. At each selected BS, a cached submodel is randomly chosen and enlarged, and a combination of the remaining submodels that satisfies the memory constraint is then randomly selected for caching.

\end{itemize}

\begin{table}[t]
    \centering
    \caption{Performance Comparison in Online Scenario (HR: Hit Rate)}
    \label{tab:online_res}
    \vspace{-6pt}
    \begin{tabular}{@{}lccccc@{}}
        \toprule
        & \multicolumn{2}{c}{\textbf{w/ Partition}} & \multicolumn{2}{c}{\textbf{w/o Partition}} \\
        \cmidrule(lr){2-3} \cmidrule(l){4-5}
        \textbf{Method} & \textbf{Avg. QoE} & \textbf{Avg. HR} & \textbf{Avg. QoE} & \textbf{Avg. HR} \\
        \midrule
        \textbf{CoCaR-OL (Ours)} & \textbf{0.780} & \textbf{0.963} & \textbf{0.414} & \textbf{0.472} \\
        LFU-MAD & 0.457 & 0.556 & 0.216 & 0.246 \\
        LFU & 0.420 & 0.511 & 0.202 & 0.229 \\
        Random & 0.415 & 0.536 & 0.098 & 0.112 \\
        \bottomrule
    \end{tabular}
    \vspace{-3pt}
\end{table}

In our evaluation, for the above algorithms, only submodels of models that are not currently being downloaded can be switched, and the memory required to download submodels is considered in each decision. Additionally, we considered versions of CoCaR-OL and the above baselines without the dynamic DNN switching mechanism (i.e., model partitioning), meaning that either the model is not cached at all or the complete, original model is cached.
\revi{As shown in Table \ref{tab:online_res}, CoCaR-OL (w/ Partition) demonstrates superior performance in the online scenario, consistently outperforms all baselines in average QoE and hit rate.}

\begin{figure}[t]
\centering
\includegraphics[width=.48\textwidth]{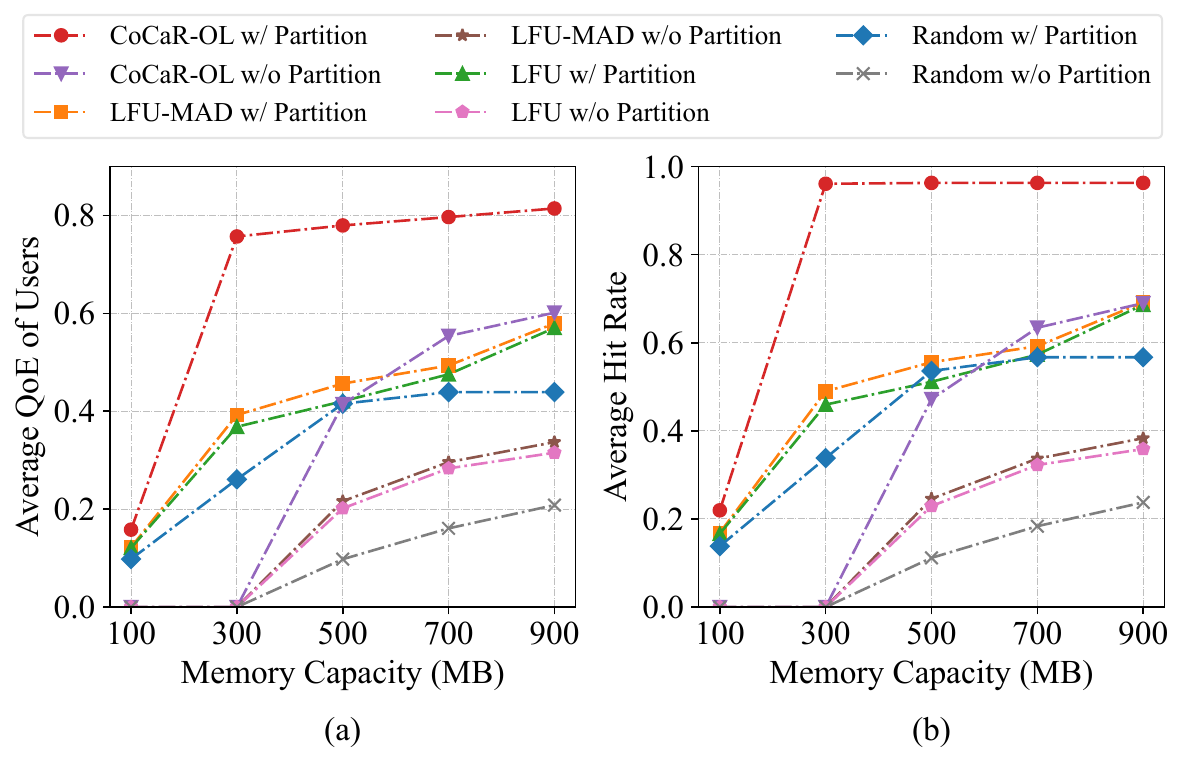}
\vspace{-6pt}
\caption{Effect of BS memory capacity: (a) Average QoE of users; (b) Average hit rate.}
\label{fig:OL_R} 
\end{figure}

\begin{figure}[t]
\centering
\includegraphics[width=.48\textwidth]{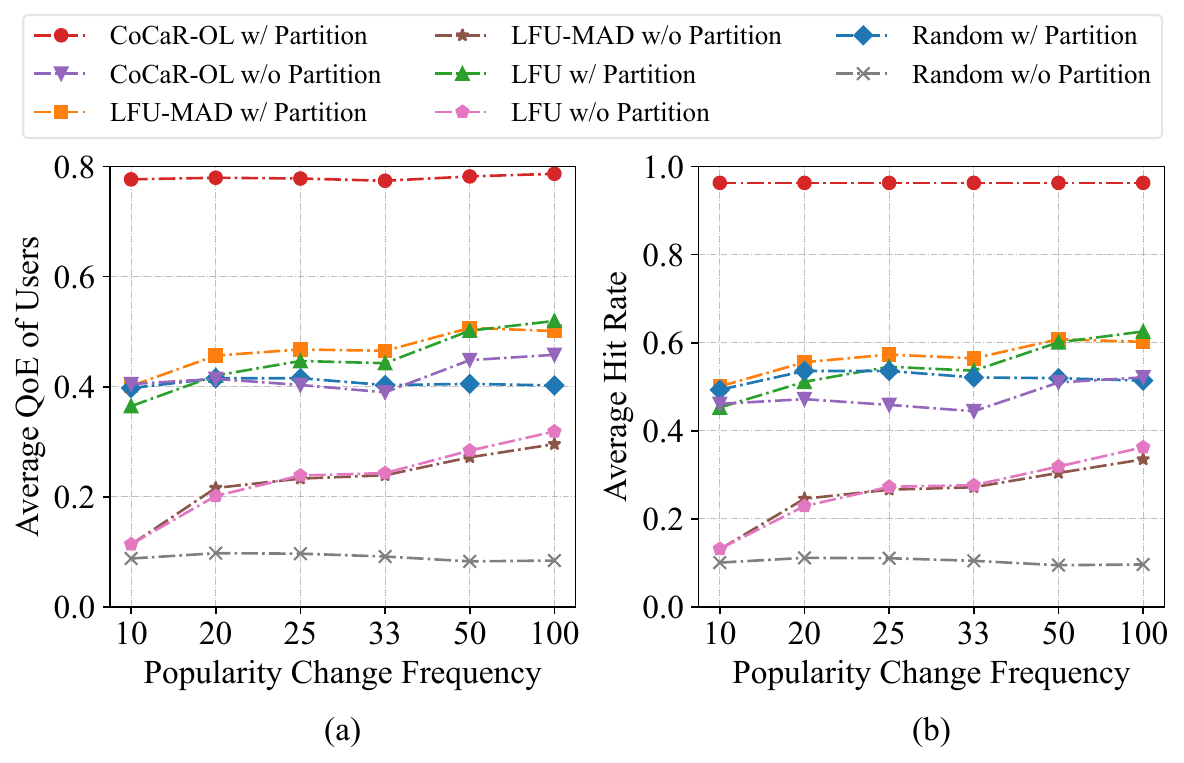}
\vspace{-6pt}
\caption{Effect of popularity change frequency: (a) Average QoE of users; (b) Average hit rate.}
\label{fig:OL_C} 
\end{figure}

\begin{figure}[t]
\centering
\includegraphics[width=.48\textwidth]{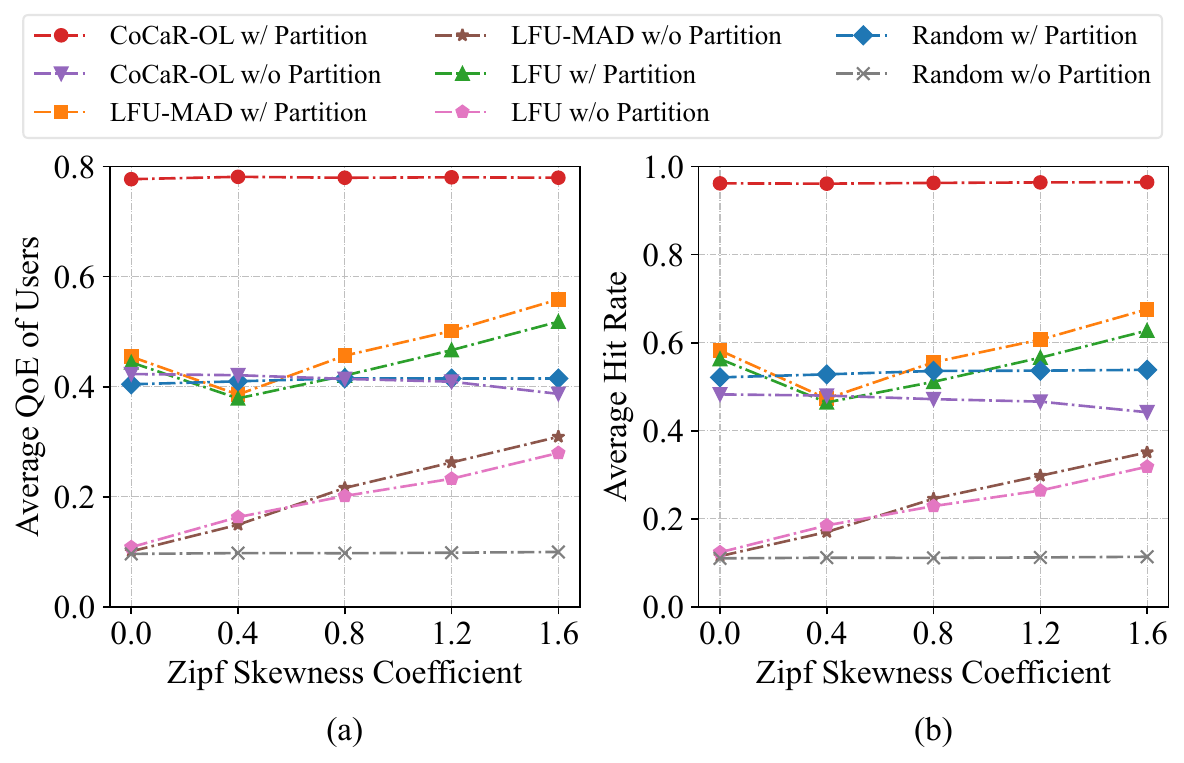}
\vspace{-6pt}
\caption{Effect of Zipf skewness coefficient: (a) Average QoE of users; (b) Average hit rate.}
\label{fig:OL_Z} 
\end{figure}

\textbf{Effect of BS Memory Capacity.}
We evaluated the performance of each algorithm with cache sizes ranging from 100MB to 900\,MB. 
As shown in Fig. \ref{fig:OL_R}, when the cache is small (e.g., 100\,MB), only the smallest submodel of each model can be cached, resulting in low user QoE and cache hit rates across all algorithms. 
As the cache size increases, larger submodels can be cached, thereby providing better service to more users and improving QoE and hit rate for all algorithms. 
Moreover, algorithms with dynamic DNN submodel switching consistently outperform their counterparts that cache only the full original models.
For instance, with a cache size of 500\,MB, CoCaR-OL with model partitioning achieves a 32.3\% higher user QoE than LFU-MAD with partitioning, and a 36.5\% higher QoE than CoCaR-OL without partitioning. 
This improvement stems from our dynamic DNN-based caching mechanism, which continues to serve users with the current submodels while new submodels are being downloaded. 
These results demonstrate that the dynamic DNN caching mechanism can adapt to different caching schemes, leveraging fine-grained memory usage and enabling flexible, rapid model switching to enhance user service.

\textbf{Effect of Popularity Change Frequency.}
Fig. \ref{fig:OL_C} illustrates the impact of the frequency of DNN model popularity changes, with the popularity changing from every 10 time slots to every 100 slots. 
Overall, the user's QoE and cache hit rate of all algorithms gradually increase as the frequency of popularity changes decreases. 
Furthermore, since the proposed CoCaR-OL algorithm considers the impact of each caching decision on the global expected future gain, while others rely solely on partial information, it exhibits greater stability and consistently outperforms other algorithms regardless of popularity change frequency.
When popularity changes every 10 time slots, LFU with model partitioning achieves 3.3\% and 3.2\% lower user QoE than LFU-MAD and Random, respectively. This is because LFU computes request frequencies over the past 10 time slots, so when popularity shifts abruptly, its caching decisions lag behind the new distribution. In contrast, LFU-MAD gives more weight to recent requests, enabling faster adaptation to local popularity changes.

\textbf{Effect of the Zipf Skewness Coefficient.}
The impact of the Zipf skewness coefficient on algorithm performance is shown in Fig. \ref{fig:OL_Z}.
As the skewness coefficient increases, the proposed CoCaR-OL consistently leverages global information to generate high-quality caching decisions, maintaining stable performance. 
In contrast, LFU and LFU-MAD with model partitioning exhibit much weaker performance.
This is because the LFU strategy relies exclusively on frequency for its decisions, disregarding other critical factors such as model size and precision.
These limitations further underscore the advantages of our gain-oriented CoCaR-OL in balancing precision and latency.
When the Zipf coefficient is zero, popular models change randomly, and LFU tends to switch to caching smaller submodels of different models to provide some services.
As the Zipf coefficient increases to 0.4, LFU gradually caches the largest submodels of popular models, but the requests for these models may not greatly exceed those for less popular models, leading to overall performance degradation. This degradation is mitigated by further increasing the coefficient, ultimately resulting in improved performance.

\section{Conclusion and future work}
\label{conclusion}

\revii{In this paper, we study the joint dynamic model caching and request routing optimization problem in MEC networks. 
By introducing dynamic DNNs, the complete original model is decomposed into multiple interrelated and switchable submodels, enabling more fine-grained and flexible resource management as well as adaptive caching strategy adjustments.}
Then, we perform multiple equivalent transformations and relaxations for the joint optimization problem and propose a novel random rounding and LP-based algorithm, CoCaR, to solve it.
Moreover, to adapt to the online scenario where user requests are difficult to predict, we further propose CoCaR-OL as an extension of CoCaR.
Theoretical analysis and experimental results show that the proposed CoCaR achieves near-optimal performance, and CoCaR-OL outperforms other benchmark algorithms in online settings.
Therefore, our proposed dynamic DNN-based caching mechanism effectively exploits memory resources in a fine-grained manner, flexibly adjusting caching strategies to enhance service delivery for users.

In the future, we will further explore BSs' resource allocation and distributed decision-making scheme in the joint model caching and request routing problem.
\revi{We will also take into account real-world environmental factors, such as network dynamics, congestion, and failures, to enhance the system's scalability and robustness in complex edge environments.}

\section*{Acknowledgments} 
The authors also thank the Big Data Computing Center of Southeast University for providing the experiment environment and computing facility.
The authors are also grateful to SF Technology Co., Ltd. for their valuable discussions and technical support.

\bibliographystyle{IEEEtran}
\bibliography{reference}

\end{CJK}
\end{document}